\documentclass[twocolumn,journal]{IEEEtran}
\IEEEoverridecommandlockouts

\newtheorem{lemma}{Lemma}

\newtheorem{theorem}{Theorem}

\usepackage{hyperref}
\usepackage{cite}
\usepackage{amsmath}
\usepackage{graphicx}
\usepackage{array}
\usepackage{multirow}
\usepackage{subfigure}
\usepackage{xcolor}
\usepackage{balance}

\usepackage{amsfonts}
\usepackage{amssymb}
\usepackage{algorithm}
\usepackage{algorithmic}

\def\bege{\begin{equation}} \def\ende{\end{equation}}

\def\begr{\begin{eqnarray}} \def\endr{\end{eqnarray}}

\def\bege{\begin{equation}} \def\ende{\end{equation}}
\def\begr{\begin{eqnarray}} \def\endr{\end{eqnarray}}
\def\bnum{\begin{enumerate}} \def\enum{\end{enumerate}}

\newcounter{mytempeqncnt}

\begin{document}
\bibliographystyle{IEEEtran}
\title{Infectious Probability Analysis on COVID-19 Spreading with Wireless Edge Networks}

\author{Xuran Li, Shuaishuai Guo,~\IEEEmembership{Member, IEEE,} Hong-Ning Dai,~\IEEEmembership{Senior Member, IEEE} and Dengwang Li
\thanks{Manuscript received January 15, 2022; revised May 1, 2022; accepted June 16, 2022.
The work is supported in part by the National Natural Science Foundation of China under Grant 62171262, in part by Shandong Provincial Natural Science Foundation under Grant ZR2021YQ47, in part by Major Scientific and Technological Innovation Project of Shandong Province under Grant 2020CXGC010109, in part by Tashan Young Scholar under Grant No. tsqn201909043, in part by the National Natural Science Foundation of China (61971271), the Jinan City-School Integration Development Strategy Project (JNSX2021023), and the Shandong Province Major Technological Innovation Project (2022CXGC010502).
(\emph{Corresponding author: Shuaishuai Guo}.)
}
\thanks{Xuran Li and Dengwang Li are with Shandong Key Laboratory of Medical Physics and Image Processing, School of Physics and Electronics, Shandong Normal University, Jinan 250061, China (e-mail: sdnulxr@sdnu.edu.cn; dengwang@sdnu.edu.cn).}
\thanks{ Shuaishuai Guo is with School of Control Science and Engineering, Shandong University, Jinan 250061, China and also with Shandong Provincial Key Laboratory of Wireless Communication Technologies (e-mail: shuaishuai\textunderscore guo@sdu.edu.cn).}
\thanks{Hong-Ning Dai is with the Department of Computer Science, Hong Kong Baptist University, Hong Kong SAR (e-mail: hndai@ieee.org).}
}


\maketitle
%
\begin{abstract}
The emergence of infectious disease COVID-19 has challenged and changed the world in an unprecedented manner. The integration of wireless networks with edge computing (namely wireless edge networks) brings opportunities to address this crisis. In this paper, we aim to investigate the prediction of the infectious probability and propose precautionary measures against COVID-19 with the assistance of wireless edge networks. Due to the availability of the recorded detention time and the density of  individuals within a wireless edge network, we propose a stochastic geometry-based method to analyze the infectious probability of individuals. The proposed method can well keep the privacy of individuals in the system since it does not require to know the location or trajectory of each individual. Moreover, we also consider three types of mobility models and the static model of individuals. Numerical results show that analytical results well match with simulation results, thereby validating the accuracy of the proposed model. Moreover, numerical results also offer many insightful implications. Thereafter, we also offer a number of countermeasures against the spread of COVID-19 based on wireless edge networks. This study lays the foundation toward predicting the infectious risk in realistic environment and points out directions in mitigating the spread of infectious diseases with the aid of wireless edge networks.


\end{abstract}

\IEEEpeerreviewmaketitle

\section{Introduction}
\label{sec:intro}
Recently, the rapid spread of the new coronavirus disease (COVID-19) has brought serious challenges to the whole world. As a high infectious disease, the virus of COVID-19 could spread from humans to humans through respiratory droplets, aerosols and other transmission manners~\cite{Wu:Nature2020,wong2022transmission}. This disease attacks the respiratory system of infected individuals and results in many symptoms, such as fever, fatigue, dry cough, muscular pain, and breathlessness. Therefore, taking effective countermeasures to combat the spread of the COVID-19 becomes an important research topic for researchers in different fields~\cite{Chang:2020nature,9405303}.

\subsection{Motivation}
Research efforts from communications and computer communities have been conducted to fight against the spread of COVID-19. In particular, recent studies~\cite{Ning:2021JSAC,Azana:2021JNCA, Pace:2019TII, Qadri:CST20} have investigated to deploy the Internet of medical things (IoMT) and establish the telemedicine platforms so as to mitigate the bottlenecks at public healthcare institutions. Meanwhile, data analytics and artificial intelligence (AI) techniques have been investigated to combat COVID-19 from different perspectives, such as diagnosis of new variants of COVID-19, spread prediction, and transmission risk analysis~\cite{Zhou:2021JSAC,Hossain:2020IN, Gomez:2021CogCom,Islam:2021IA,ren2022optimal}. Moreover, blockchain techniques have been adopted for contact tracing with privacy preservation~\cite{LvW:2020TNCE, xrli:2021PMC,Fakhri2020:SCS}. However, most of the aforementioned methods have stringent demands on the network performance, such as low latency, high bitrate, and the ability to handle requests from a large number of devices. These proposed methods may not be feasible for off-the-shelf network/computing infrastructures. 

To fulfill the increasing computational/storage demands in IoMT and telemedicine systems, a typical solution is to outsource computational-complex tasks to remote cloud service providers, which have stronger computational capability than end devices. However, those cloud service providers are often owned by untrusted third parties, who may mistakenly or intentionally breach the data privacy~\cite{9665215,10.1145/3425707}. Meanwhile, uploading computing tasks to remote clouds inevitably leads to high latency. As an important complement to cloud computing, edge computing has recently appeared as a promising solution to IoMT~\cite{9320508,Ning:2021JSAC} since some computational tasks can be offloaded to nearby edge computing nodes. Edge computing nodes are typically deployed with existing wireless infrastructures (such as macro base stations, small base stations, access points, and IoMT gateways) to form \emph{wireless edge networks}. In this paper, we investigate the adoption of wireless edge networks to analyze and combat the spread of COVID-19.

The spread prediction of COVID-19 is one of the most important countermeasures against the viral outbreaks~\cite{TOMAR2020138762}. For example, heat warning can provide a rough estimation on possible infectious people by analyzing thermal images~\cite{MORABITO2020140347}. Moreover, the prediction system based on the confirmed infection cases~\cite{10.1145/3397271.3401429} can predict hazard areas. Despite the advent of these studies, most of them rely on data analysis at cloud servers, which have privacy-leakage risks (as analyzed above). In addition, the timely warning is crucial for an early warning and disease control while cloud services inevitably cause high latency. To address these issues, offloading the prediction tasks at edge nodes is a promising solution. To the best of our knowledge, there is no study on spread prediction of COVID-19 and investigation of countermeasures based on wireless edge networks.

\subsection{Contributions}
In this paper, we focus on establishing an analytical framework to analyze the infectious probability of susceptible individuals and providing early warnings by exploiting wireless edge networks. As indicated in previous studies~\cite{yan2008distribution,ZHANG2020201,LIU2021106542}, the infectious probability of a susceptible individual heavily depends on both the contact time and the distance between the infected individual and the desired susceptible individual. The detention time of individuals in a network is essentially available at an edge server by service providers (i.e., we use the detention time, which is longer than the contact time, to calculate the infectious risk). Moreover, the distance can be modelled by stochastic mechanism and stochastic geometry~\cite{yan2008distribution,Andrews:2011Tcom,Hmamouche:2021IP,Lu:2021CST}. Inspired by these previous findings, we then establish an analytical framework to evaluate the infectious probability of susceptible individuals. Moreover, we consider the impact of the mobility of individuals into our analytical framework. In particular, our framework consider three types of mobility models: the random direction (RD) model~\cite{Govindan:2011TWC,Tang:2020IoT}, random walk model (RWK)~\cite{XJunfei:2014ICST,GZhenhua:2011ICC} and random waypoint (RWP) model~\cite{TJie:2018TWC,Valentine:2020IA,Fernandez:2018TVT}. Unlike other methods requiring the location or trajectory of each individual, our method can better protect the privacy of individuals since only the recorded detention time\footnote{The detection time or access time can be obtained by mobile service providers while less privacy is leaked in contrast to other pandemic surveillance methods, which require to access more user-privacy sensitive data~\cite{RIBEIRONAVARRETE2021120681}.} is used in our analysis.
After establishing the analytical framework of the infectious probability and evaluating the impacts of multiple factors on the infectious probability by extensive simulations, we also discuss the countermeasures to combat the spread of infectious diseases like COVID-19 and its variants. For example, edge servers can offer early warnings to individuals within the same network as the infected individual so that the corresponding countermeasures (e.g., keeping social distance, wearing masks, and improving ventilation) can be done.

The main research contributions of this paper can be summarized as follows.
\begin{itemize}
\item In this work, we established a novel analytical framework to analyze the infectious probability of susceptible
individuals within wireless edge networks. Deploying this analytical framework at wireless edge networks can potentially provide individuals with early warnings in time while leaking less privacy of individuals.

\item This analytical framework also investigates the effect of individual mobility on infectious probability. Specifically, the proposed framework considers the three most commonly used mobility models: the RD model, RWK model and RWP model. Extensive simulation results agree with the analytical results, implying the accuracy of the proposed framework.

\item Analytical results suggest a number of countermeasures, such as early warning and social distancing. The integration of these countermeasures with wireless edge networks can effectively mitigate the the spread of infectious diseases.

\end{itemize}

The remainder of the paper is organized as follows. Section~\ref{sec:pre} introduces the network model, the mobility models of individuals, and the infectious model. Section \ref{IPA} presents the infectious probability analysis, including the analysis of static individuals and the impact of mobility on the infectious probability analysis. Section \ref{sec:result} demonstrates the simulation results. Section \ref{sec:conclusion} concludes this paper. Section \ref{sec:future} pointed out the future research directions.

\section{System Model}
\label{sec:pre}

This section presents the models including the network model, infectious model and mobile models in Sections~\ref{subsec:network},~\ref{subsec:infectious}~\ref{subsec:mobility}, respectively.


\subsection{Network Model}\label{subsec:network}

Consider a wireless edge network as illustrated in Fig.~\ref{fig:network}, where a number of mobile users are connected with a base station, e.g., macro base station, Evolved Node B (eNB) in 4G LTE or gNode in 5G networks. The base station is equipped with an edge server, which can provide data storage and computation services for mobile users. In this paper, we also exploit the edge server to analyze the infectious probability and send early warnings.

\begin{figure}[t]
\centering
\includegraphics[width=9cm]{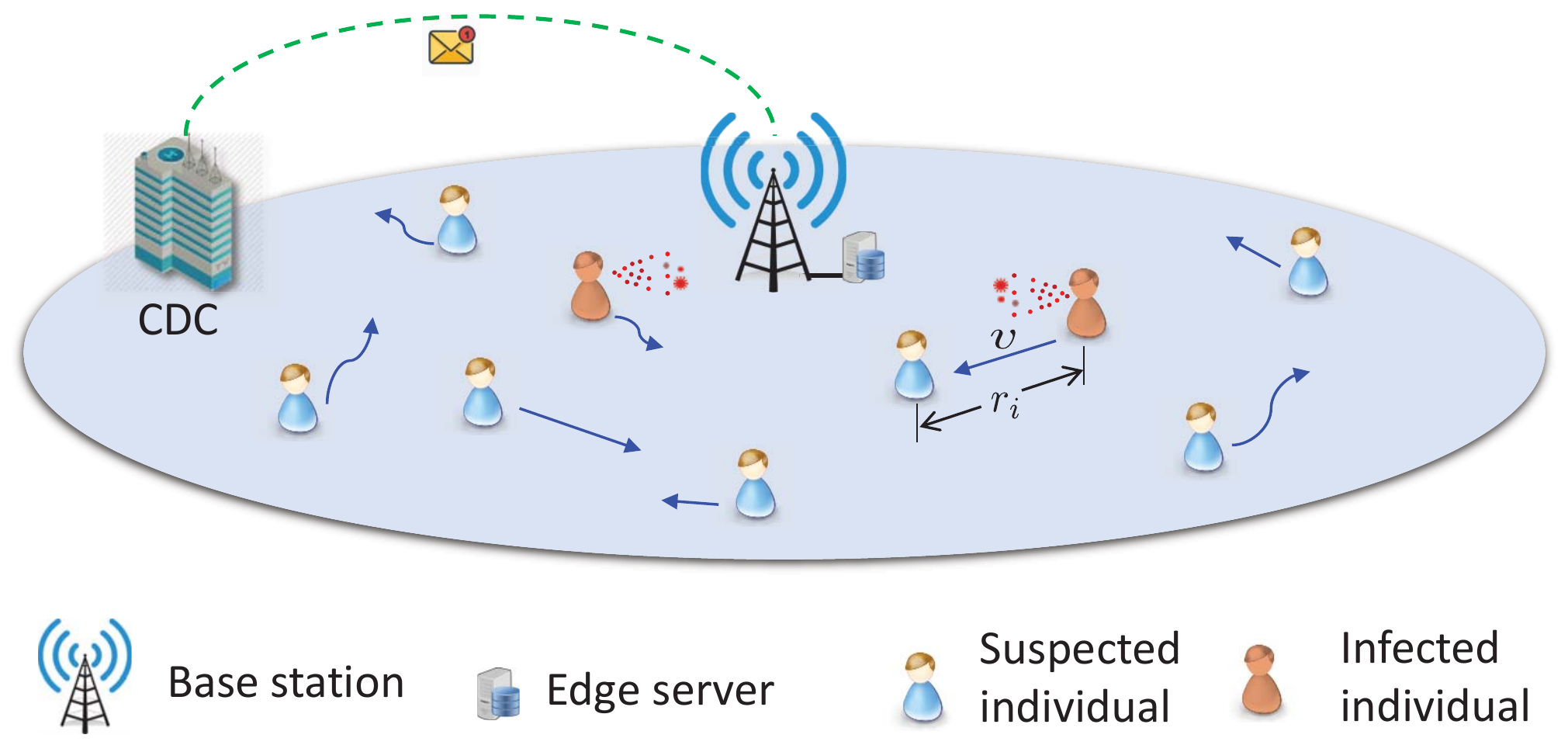}
\caption{System Model of Infectious Probability Analysis in Wireless Edge Network.}
\label{fig:network}
\end{figure}

There are two types of mobile users randomly distributed in this network. (1) The \emph{infected individuals} are those individuals who have already been infected by the virus and may transmit the virus due to the high virus volume. (2) The \emph{susceptible individuals} are those individuals who are not infected and can nevertheless become infected by infected individuals due to the exposure to the infected individuals (especially in poorly ventilated environment or enclosed environment). Section~\ref{subsec:infectious} will present the infectious model. Both infected individuals and susceptible individuals are randomly distributed according to the uniform distribution. Moreover, those users are moving within this network according to different mobile patterns (they are modelled according to three types of mobility models in Section~\ref{subsec:mobility}). Note that both the number of infected individuals and the number of susceptible individuals are time-varying since some newly infected individuals (or susceptible individuals) may join the network while some of them may leave in this network due to the mobility of individuals. In addition, the susceptible individuals may become the infected individuals after being medical diagnosed or tested (e.g., nucleic acid tests).

The edge server can analyze the infectious risk based on the following available information. i)
The statistical characteristics of individuals' distribution  (i.e., the location distribution of individuals) in the network area, can be obtained from the traffic management system or other related departments (or services). ii) Infected individuals can be obtained by Centers for Disease Control and Prevention (CDC) or other agencies though the privacy of the infected individuals can be properly protected by pseudonymity or other cryptographic schemes (i.e., hiding the exact individual identification). iii) The detention time of each individual (both infected individual and susceptible individual) is available by mobile services providers. With the availability of the above information, we aim to establish an analytical framework to analyze the infectious probability of each susceptible individual.
In our framework, both specific location and trajectory of each individual are not used. Therefore, the risk of privacy disclosure for each individual is avoided.
Without loss of generality, the analysis of the infectious probability is conducted by assuming the reference susceptible individual to be located at the center. Based on the analytical model, the edge server can calculate the infectious probability of each individual and send early warning messages to the individuals (or suggest the corresponding countermeasures in Section~\ref{sec:future}).

\subsection{Infectious Model}\label{subsec:infectious}
Although there are a number of studies on investigating the transmission of infectious diseases~\cite{yan2008distribution,ZHANG2020201,LIU2021106542,SUN2020102390}, they are too complex to be directly adopted for the spread prediction of infectious diseases in the scenarios of wireless edge networks. Thus, it is a necessity to establish a simple yet effective infectious model for estimating infectious risks in a crowded scenario with the help of wireless edge networks.

As indicated by previous studies~\cite{yan2008distribution,LIU2021106542,SUN2020102390,SINGANAYAGAM2022183}, the infectious risk significantly drops with the increased distance between the infected individual and a susceptible individual. Inspired by these findings, we propose a simple-yet-effective infectious model with consideration of multiple factors (such as social distance, respiratory droplets, virus volume, time, and transmission factor of infectious virus). In this model, we first define the metric of instant infectious strength denoted by $I_\text{inf}$ to evaluate the instantaneous spreading volume of the virus transmitted from infected individuals to the susceptible individual within a unit of time because previous studies indicated the positive relation between the volume of virus and the infectious rate~\cite{LIU2021106542,SUN2020102390,SINGANAYAGAM2022183}. Specifically, $I_\text{inf}$ is given as follows,
\begin{equation}
I_\text{inf}= \sum_{i=1}^N V_{i} \cdot r_{i}^{-\eta},
\label{eq:Iinf}
\end{equation}
where $N$ is the number of infected individuals, $V_{i}$ is the virus volume generated by the $i$th infected individual in the unit time period, $r_{i}$ is the distance from the $i$th infected individual to the susceptible individual, and $\eta$ is the path loss factor of the virus spreading that varies from 2 to 7~\cite{LIU2021106542}.Note that the virus volume is a variable that varies with different infected individuals~\cite{he2020temporal}. For example, a higher virus volume may be generated by infected individuals whose have more serious symptoms while the infected individuals with slight symptoms may generate fewer viral particles and have a smaller virus volume~\cite{he2020temporal,SINGANAYAGAM2022183}. We denote the maximal virus volume and the minimal virus volume generated by an infected individual by $V_{M}$ and $V_{m}$, respectively. Therefore, the virus volume $V_{i}$ generated by the $i$th infected individual varies from $V_{m}$ to $V_{M}$.

As shown in previous studies~\cite{LIU2021106542,SUN2020102390}, a higher virus volume leads to a higher chance of an individual being affected and a longer detention time (or contact time) leads to a higher risk of an individual being affected. We let $V_\text{th}$ denote the threshold virus volume that may lead a susceptible individual to become an infected individual. We then define the instantaneous infectious probability to evaluate the probability that a susceptible individual may become an infected individual within a unit of time. The instantaneous infectious probability is denoted by $\mathbb{P}_\text{inf}$, which is expressed as follows,
\begin{equation}
\mathbb{P}_\text{inf} =  P( I_\text{inf}\ge V_\text{th}) = P \left(  \sum_{i=1}^N V_{i} \cdot r_{i}^{-\eta} \ge V_\text{th} \right).
\label{eq:pinf}
\end{equation}

We next define the total risk of a susceptible individual becoming infected with consideration of the detention time. The total risk of a susceptible individual becoming infected is denoted by $\mathbb{R}_\text{total}$. Since the detention duration $T$ that a susceptible individual is available at the edge server, we can calculate the total risk $\mathbb{R}_\text{total}$  as follows,
\begin{equation}
\mathbb{R}_\text{total}= \int_{0}^{T} \mathbb{P}_\text{inf}  dt= \int_{0}^{T} P \left(  \sum_{i=1}^N V_{i} \cdot r_{i}^{-\eta} \ge V_\text{th} \right) dt.
\label{eq:Rtotal}
\end{equation}

\subsection{Mobility Models}\label{subsec:mobility}

The mobility of an infected individual may affect the infectious probability of the susceptible individual due to the varied distance. Our analytical framework also considers three conventional mobility models, which are given as follows.

\subsubsection{RD Model} Each infected individual moves toward a random direction within $[0, 2 \pi)$, which is independent of other infected individuals. Meanwhile, the infected individual moves a random distance $R$ toward this direction at a constant speed $v$. Upon his/her arrival of a certain location, another random direction and distance $R$ are chosen; this procedure repeats.

\subsubsection{RWK Model} This model was originally proposed to describe the unpredictable movement of particles in nature~\cite{XJunfei:2014ICST}. In our framework, each infected individual moves toward a random direction within $[0, 2 \pi)$, independently of the other infected individuals, and moves with a fixed distance $W$ toward this direction at a constant speed $v$. Similarly, another random direction is also chosen when the infected individual reaches a location. This procedure repeats the above steps.

\subsubsection{RWP Model} In the beginning, each infected individual stops for a random
time $T$ at his/her initial location. He/she then moves toward a random direction within $[0, 2 \pi)$, independently of other infected individuals, and moves with a random distance $R$ toward this direction at a constant speed $v$. Upon his/her arrival, he/she stops for another random time $T$, then selects another random direction and distance $R'$. This procedure repeats.
During the procedures, both the speed $v$ and the pause time $T$ are random variables, which can be sampled independently from distributions $f_{v}(v)$ and $f_{T}(T)$, respectively. The distribution functions of $f_{v}(v)$ and $f_{T}(T)$ are usually assumed to be uniform distribution functions~\cite{Bandyopadhyay:2007TMC,Fernandez:2018TVT,Valentine:2020IA} though the authors in~\cite{Hyytia:2006TMC} mentioned that any distribution can be used (e.g., the beta distribution and discrete distributions).


\section{Infectious Probability Analysis}\label{IPA}
In this section, we focus on analyzing the infectious probability and deriving its closed-form  expression. 

Despite the advent of previous studies~\cite{Martin:journalsG09,Gong:2014TMC,Irio:2018TCOM}, there is no exact analytical expression for the probability density function (PDF) of instant infectious strength $I_\text{inf}$, which nevertheless is crucial for deriving the infectious probability. To address this issue, we separately analyze the virus volume from the dominant infected individual and that from minor infected individuals, inspired by previous work~\cite{Chetlur:2017TCOM}. The dominant infected individual is the nearest infected individual to the susceptible individual while the other infected individuals are minor infected individuals since airborne and droplet-borne infections are typically limited by distance~\cite{SUN2020102390}. 
Based on this consideration, we mainly investigate the effect of dominant infected individuals and approximate the aggregated impact from the rest of the infected individuals (namely minor infected individuals) to a Gaussian random variable. We first present an analysis on the infectious probability with static individuals in Section~\ref{subsec:static} and then extend our analysis to the scenarios of mobile individuals in Section~\ref{subsec:mobi-prob}.

\subsection{Infectious Probability with Static Individuals}\label{subsec:static}

According the infectious model in Eq.~\eqref{eq:Iinf}, we find that the infectious probability is mainly determined by two variables: the virus volume of the infected individual, and the distance between the susceptible individual and the infected individual.
The virus volume of infected individuals follows the uniform distribution within $[V_{m}, V_{M}]$ (as given in Section~\ref{subsec:infectious}). We then have the probability density function of virus volume $V_{i}$ as follows,
\begin{equation}
{f_V}(V_{i}) = \frac{1}{V_{M}-V_{m}}  , V_{m} < V_{i} \le V_{M}.
\label{eq:fv}
\end{equation}

Compared to the virus volume generated by an infected individual, the distance between the reference susceptible individual and an infected individual $r$ has a more obvious impact on the infection process~\cite{SUN2020102390}. Since the reference susceptible individual is located at the center of circular area and each infected individual follows  independent and identically distributed (i.i.d.) uniform distribution, the probability density function of $r$ is given by~\cite{Alouini:1999TVT},
\begin{equation}
{f_R}(r) = \frac{2\pi r}{\pi {D^2}} = \frac{2r}{D^2} ,~ 0 < r \le D.
\label{eq:fr}
\end{equation}

We next have the following lemmas.
\begin{lemma}
When infected individuals are static,
the distance between the susceptible individual and the dominant infected individual $R_{1}$ is given by
\begin{equation}
f_{R_{1}}\left(r_{1}\right)=\frac{2Nr_{1}(D^2-r_{1}^2)^{N-1}}{D^{2N}}.
\label{eq:fR1}
\end{equation}
\label{lemma:FR1_s}
\end{lemma}

\begin{proof}
Since the dominant infected individual is the closest infected individual to the susceptible individual, the conditional cumulative distribution function (CDF) of $R{_1}$ can be computed by the approach in~\cite{Chetlur:2017TCOM} as follows,
\begin{equation}
\begin{aligned}
F_{R_{1}}(r_{1}) &=\mathbb{P}\left(r_{1}\leq r \right) =1-\mathbb{P}\left(\min \left\{R_{N}\right\}> r \right) \\
&=1-\mathbb{P}\left(R_{1}>r, R_{2}>r, \ldots, R_{N}>r \right) \\
& \stackrel{(a)}{=} 1-\left(1-F_{R}(r)\right)^{N},
\end{aligned}
\end{equation}
where $(a)$ follows the i.i.d. nature of the set of distances $R_{N}$. Differentiating the above expression with respect to $r$, PDF of $R{_1}$ is given by
\begin{equation}
\begin{split}
f_{R_{1}}\left(r_{1}\right)&=N\left(1-F_{R}\left(r_{1} \right)\right)^{N-1} f_{R}\left(r_{1} \right)\\
&=\frac{2Nr_{1}(D^2-r_{1}^2)^{N-1}}{D^{2N}}.
\end{split}
\end{equation}
\end{proof}

For the susceptible individual situated at the origin, the PDF of the distances between the susceptible individual and the minor infected individuals \{$R_{i}$\}
conditioned on $R_{1}$ is given by
\begin{equation}
f_{R_{i}}\left(r_{i} \mid r_{1}\right)= \frac{2 r_{i}}{D^{2}-r_{1}^{2}},  r_{1} \leq r_{i} \leq D.
\label{eq:fRi_s}
\end{equation}

\begin{lemma}
When the infected individuals are static, the mean of the virus volume (excluding the virus volume from the dominant infected individual) conditioned on the distance between the susceptible individual and the dominant infected individual $R_{1}$ is
\begin{equation}
\begin{aligned}
\mu_{I_{N-1}} =\frac{(N-1) (V_{m}+V_{M}) (D^{2-\eta}-u_{1}^{2-\eta})}{(2-\eta)\left(D^{2}-u_{1}^{2}\right)}.
\label{eq:EIN1_s}
\end{aligned}
\end{equation}
\label{lemma:EIN1_s}
\end{lemma}

\begin{figure*}[ht]
\centering
\subfigure[Trails with the RD model]{
\label{fig:Trails-a}
\includegraphics[width=5.8cm]{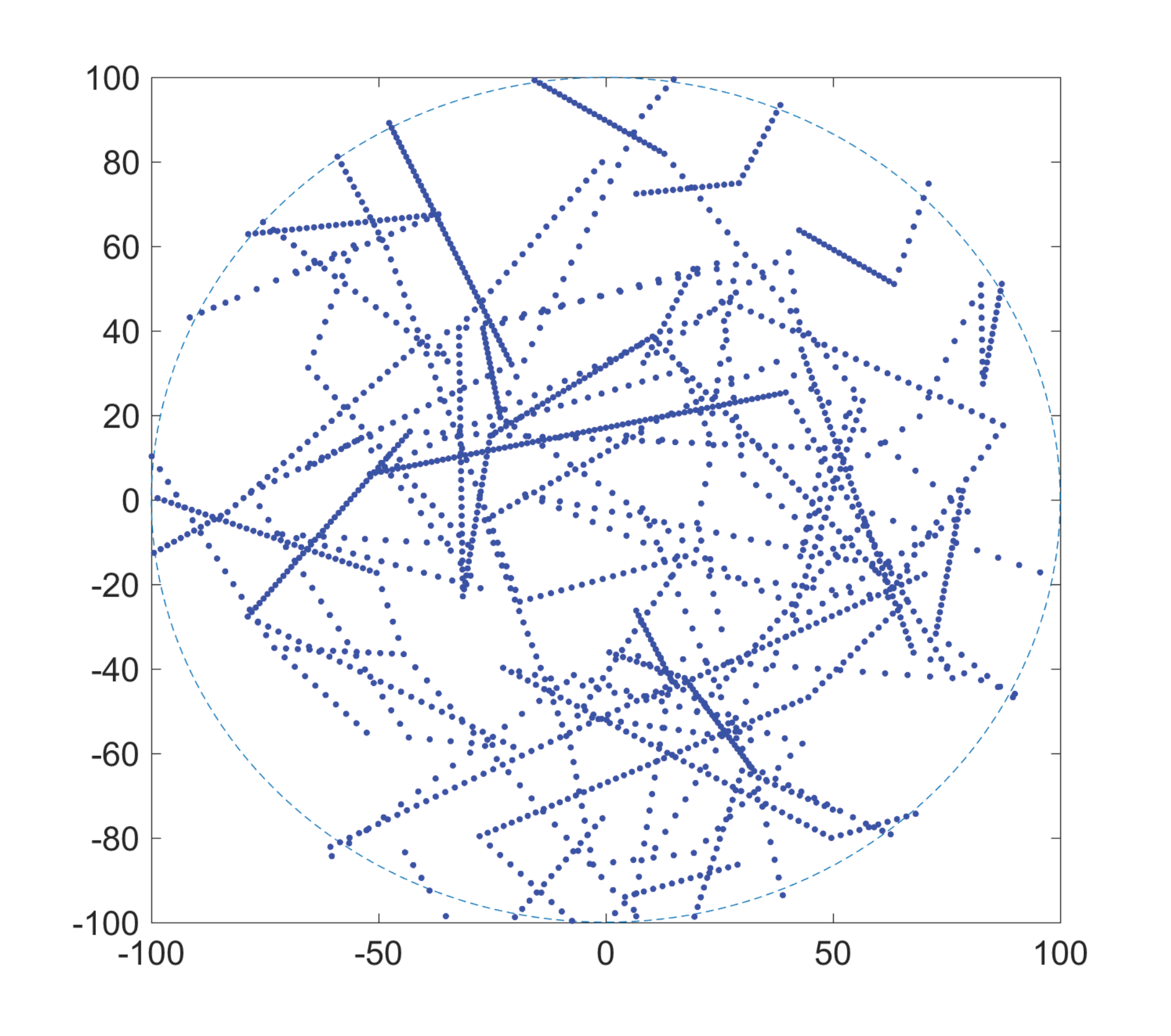}}\hfil
\subfigure[Trails with the RWK model]{
\label{fig:Trails-b}
\includegraphics[width=5.8cm]{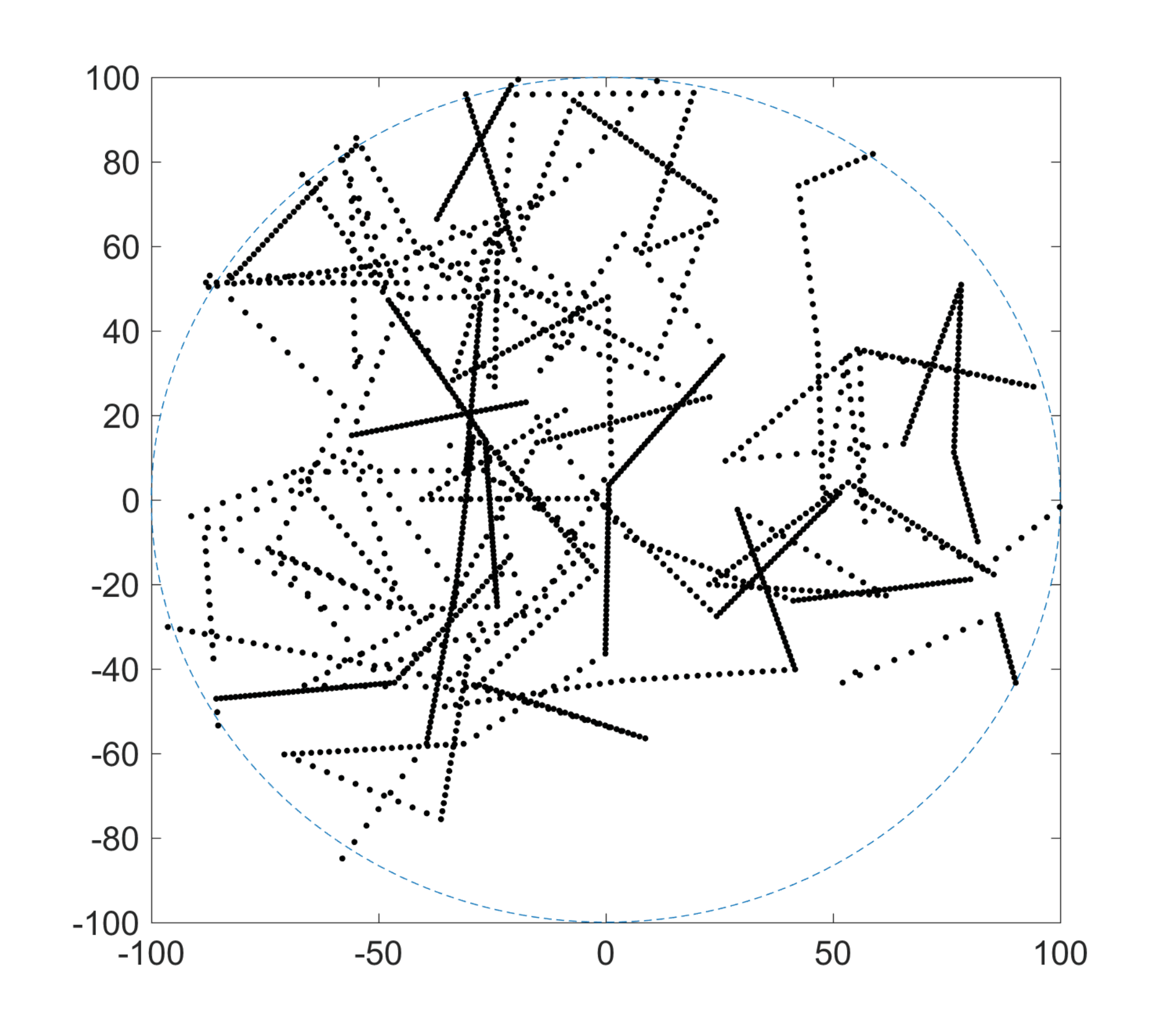}}\hfil
\subfigure[Trails with the RWP model]{
\label{fig:Trails-c}
\includegraphics[width=5.8cm]{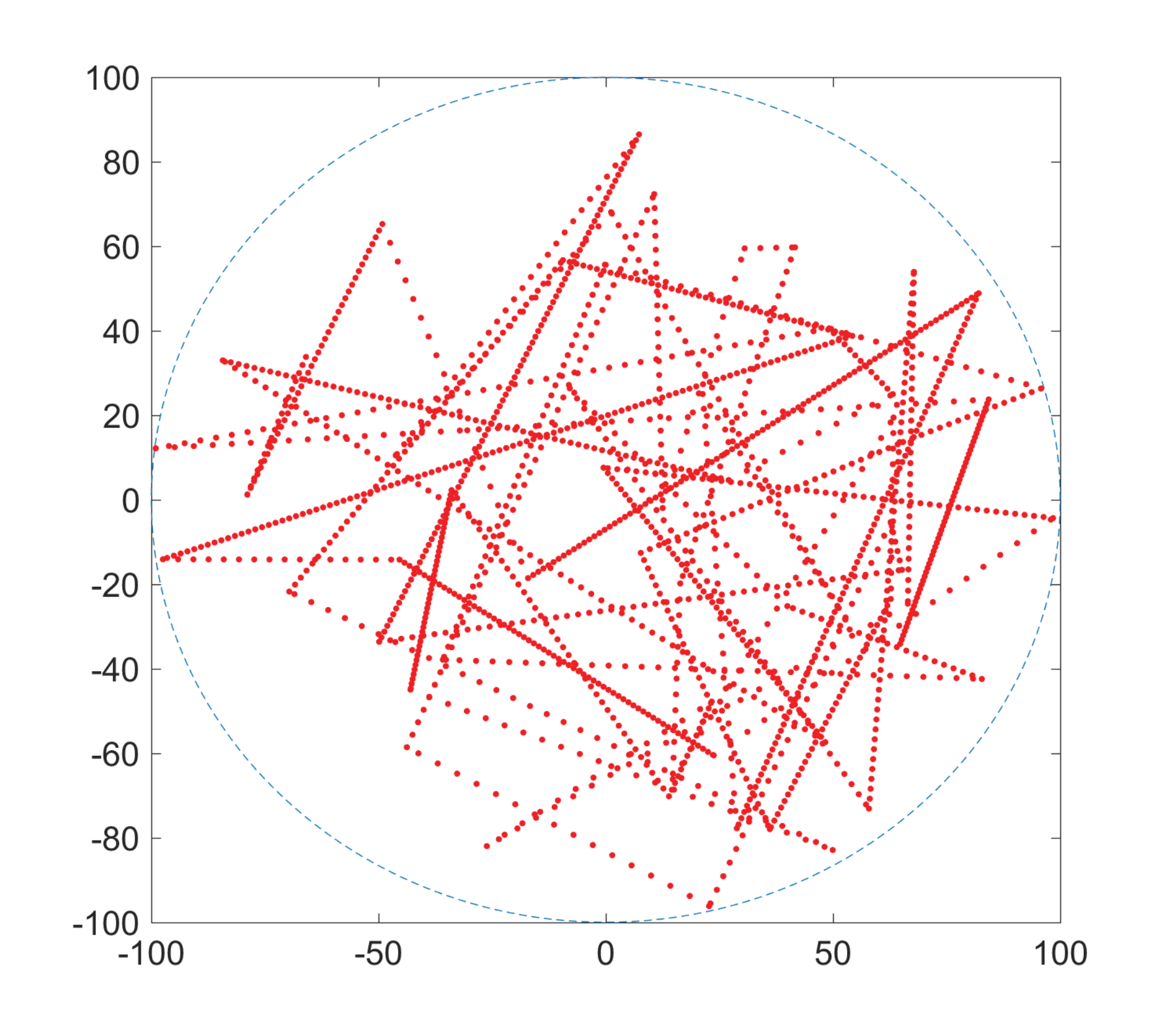}}\hfil
\caption{Trails of an infected individual moves in a circular region following the RD model, RWK model and RWP model}
\label{fig:Trails}
\end{figure*}

 \begin{proof}
According to the definition of the mean value of a variable, we have the mean of the virus volume generated by minor infected individuals as follows,
{\small
\begin{equation}
\begin{aligned}
\mu_{I_{N-1}} = & \mathbb{E}\left[I_{N-1} \mid R_{1}, V_{1}\right]\\
\stackrel{(a)}{=}&(N-1) \mathbb{E}\left[V_{i} \cdot R_{i}^{-\eta} \mid R_{1}, V_{1}\right]\\
= &(N-1) \int_{V_{m}}^{V_{M}}   \int_{r_{1}}^{D}  V_{i} \cdot  r_{i}^{-\eta}   f_{R_{i}}  \left(r_{i} \mid r_{1}\right) {f_V}(V_{i}) \mathrm{d} r_{i}  \mathrm{d} V_{i}\\
\stackrel{(b)}{=}&(N-1) \int_{V_{m}}^{V_{M}} V_{i} {f_V}(V_{i}) \mathrm{d} V_{i} \cdot \int_{r_{1}}^{D}     r_{i}^{-\eta} f_{R_{i}}  \left(r_{i} \mid r_{1}\right) \mathrm{d} r_{i},
\label{eq:mu}
\end{aligned}
\end{equation}
}

\noindent where $(a)$ follows the conditional i.i.d. nature of the distances \{$R_{i}$\} and the i.i.d. nature of the virus volumes \{$V_{i}$\}, $(b)$ follows the independence between \{$R_{i}$\} and \{$V_{i}$\}.
Substituting (\ref{eq:fv}) and (\ref{eq:fRi_s}) into Eq.~\eqref{eq:mu}, we have
{\small
\begin{equation}
\begin{aligned}
\mu_{I_{N-1}}
=&  (N-1) \int_{V_{m}}^{V_{M}}  \frac{ V_{i} } {V_{M}-V_{m}} \mathrm{d} V_{i} \int_{r_{1}}^{D} \frac{ 2 r_{i}^{1-\eta}} {D^{2}-r_{1}^{2}} \mathrm{d} r_{i}\\
=& \frac{(N-1) (V_{m}+V_{M}) (D^{2-\eta}-u_{1}^{2-\eta})}{(2-\eta)\left(D^{2}-u_{1}^{2}\right)}.
\end{aligned}
\end{equation}
}

\end{proof}

\begin{lemma}
When the infected individuals are static, the variance of virus volume (excluding virus volume generated by the dominant infected individual) conditioned on the distance between the susceptible individual and the dominant infected individual $U_{1}$ is
{\small
\begin{equation}
\begin{aligned}
{\sigma_{I_{N-1}}}
&=(N-1)\left[ \frac{ (V_{m}^2+V_{m}V_{M}+ V_{M}^2) (D^{2-2\eta}-r_{1}^{2-2\eta})}{3(1-\eta)\left(D^{2}-r_{1}^{2}\right)} \right.\\&\qquad\qquad\qquad\qquad\left.-   \frac{(V_{m}+ V_{M})^2 (D^{2-\eta}-r_{1}^{2-\eta})^2} {(2-\eta)^2 \left(D^{2}-r_{1}^{2}\right)^2}  \right].
\label{eq:VIN1_s}
\end{aligned}
\end{equation}
}

\label{lemma:VIN1_s}
\end{lemma}

\begin{proof}
From the definition of conditional variance, we can calculate the conditional variance of virus volume at the susceptible individual excluding the virus volume from the dominant infected individual as follows,

{\small
\begin{equation}\label{eq13}
\begin{aligned}
{\sigma_{I_{N-1}}}&=\operatorname{Var}\left[I_{N-1} \mid R_{1}, V_{1}\right] \\
&=(N-1)\bigg[\int_{V_{m}}^{V_{M}}  \int_{r_{1}}^{D}  {(V_{i} r_{i}^{- \eta})}^2 {f_V}(V_{i}) f_{R_{i}}\left(r_{i} \mid  r_{1}\right) \mathrm{d} r_{i}\mathrm{d} V_{i}  \\
&\quad\quad\quad-\left( \int_{V_{m}}^{V_{M}}   \int_{r_{1}}^{D}  V_{i} \cdot  r_{i}^{-\eta} {f_V}(V_{i})  f_{R_{i}}  \left(r_{i} \mid r_{1}\right) \mathrm{d} r_{i}  \mathrm{d} V_{i}  \right)^{2}\bigg].
\end{aligned}
\end{equation}
}

Substituting (\ref{eq:fv}) and (\ref{eq:fRi_s}) into (\ref{eq13}), we derive
{\small
\begin{equation}
\begin{aligned}
{\sigma_{I_{N-1}}}=&(N-1) \bigg[   \int_{V_{m}}^{V_{M}}  \frac{ {V_{i}}^2 } {V_{M}-V_{m}} \mathrm{d} V_{i}     \int_{r_{1}}^{D} \frac{ 2 r_{i}^{1-2\eta}} {D^{2}-r_{1}^{2}} \mathrm{d} r_{i}\\
&\quad\quad\quad\quad - \left( \int_{V_{m}}^{V_{M}}  \frac{ {V_{i}} } {V_{M}-V_{m}} \mathrm{d} V_{i}  \int_{r_{1}}^{D} \frac{ 2 r_{i}^{1-\eta}} {D^{2}-r_{1}^{2}} \mathrm{d} r_{i} \right)^2  \bigg]\\
= &(N-1)\left[ \frac{ (V_{m}^2+V_{m}V_{M}+ V_{M}^2) (D^{2-2\eta}-r_{1}^{2-2\eta})}{3(1-\eta)\left(D^{2}-r_{1}^{2}\right)} \right.\\&\quad\quad\quad\quad\left.-   \frac{(V_{m}+ V_{M})^2 (D^{2-\eta}-r_{1}^{2-\eta})^2} {(2-\eta)^2 \left(D^{2}-r_{1}^{2}\right)^2}  \right].\\
\end{aligned}
\end{equation}
}
\end{proof}

We then derive the infectious probability of static infected individuals as follows.
\begin{theorem}
\label{theorem:Pinf_s}
When the infected individuals are static, the infectious probability of a susceptible individual can be expressed by
\begin{equation}
\begin{aligned}
\mathbb{P}_\text{inf} \approx   \int_{V_{m}}^{V_{M}} \int_{0}^{D}Q \left(\frac{V_\text{th} - V_{1}\cdot r_{1}^{-\alpha}-\mu_{I_{N-1}}}{\sigma_{I_{N-1}}}\right)\\
 \times \frac{2Nr_{1}(D^2-r_{1}^2)^{N-1}}{D^{2N} (V_{M}-V_{m})} \mathrm{d} r_{1} \mathrm{d} V_{1}.
\end{aligned}
\end{equation}
\end{theorem}
where $Q(\cdot)$ is the $Q$ function with the expression $Q(x)=\int_{x}^{+\infty} \frac{1}{\sqrt{2 \pi}} \exp \left(-\frac{1}{2} t^{2}\right) d t$.

\begin{proof}
The infectious probability of a susceptible individual can be computed by
\begin{equation}
\begin{aligned}
\mathbb{P}_\text{inf} =  \int_{V_{m}}^{V_{M}} \int_{0}^{D} P\left(I_\text{inf}>V_\text{th} \mid  R_{1}, V_{1}\right) & f_{R_{1}}( r_{1})  {f_V}(V_{1}) \mathrm{d} r_{1}  \mathrm{d} V_{1}.
\label{eq:Pinf1}
\end{aligned}
\end{equation}
As shown in (\ref{eq:Iinf}), when the virus volume generated by a infected individual $V_{i}$ is assumed to be uniformly distributed,  $I_\text{inf}$ becomes the sum of i.i.d. random variables.
When the central limited theorem (CLT) is applied, the infectious probability can be computed by
\begin{equation}
\begin{aligned}
P\left(I_\text{inf}>V_\text{th} \mid  R_{1}, V_{1} \right)
=& P\left(\sum_{i \in \Phi } V_{i} \cdot r_{i}^{-\eta} >V_\text{th} \mid  R_{1}, V_{1}\right) \\
=& P\left( I_{N-1} >V_\text{th} - V_{1} \cdot r_{1}^{-\eta} \right) \\
\approx & Q\left(\frac{V_\text{th}- V_{1} \cdot r_{1}^{-\eta}-\mu_{I_{N-1}}}{\sigma_{I_{N-1}}}\right).
\label{eq:Pinf2}
\end{aligned}
\end{equation}

Substituting (\ref{eq:fv}), (\ref{eq:fR1}) and (\ref{eq:Pinf2}) into (\ref{eq:Pinf1}), we get the expression of $\mathbb{P}_\text{inf}$ as given above.
\end{proof}

With the expression of $\mathbb{P}_\text{inf}$, we can calculate the total risk $\mathbb{R}_\text{total}$ of a susceptible individual becoming infected after duration $T$ in the network by substituting $\mathbb{P}_\text{inf}$ into Eq.~\eqref{eq:Rtotal}.

\subsection{Effect of Mobility on Infectious Probability}\label{subsec:mobi-prob}
Considering the scenarios that infected individuals move according to three different mobility models, we further analyze the infectious probability.

\begin{figure*}[ht]
\normalsize
\setcounter{mytempeqncnt}{\value{equation}}
\setcounter{equation}{18}
\begin{scriptsize}
\begin{equation}
f_{L}(l)=\left\{\begin{array}{l}
 \displaystyle  \frac{2 l}{\pi D^{2}} \bigg[\arctan \left(\frac{l}{\sqrt{4 W^2-l^2}}\right)-\arctan\left(\frac{l^2-3 W^2}{\sqrt{10 l^2 W^2 -l^4 -9 W^4}}\right)+4 \arccos \left(\frac{l}{4 W}+\frac{3 W}{4 l}\right)\\
  \displaystyle \qquad\qquad\qquad\qquad\qquad\qquad\qquad\qquad\qquad\qquad\qquad\qquad\qquad\qquad\qquad\qquad\qquad\qquad\qquad\qquad\quad -\arccos\left(\frac{l}{2 W}\right)\bigg], \ 0 \leq l < W, W \leq z < D \\
 \qquad \qquad\qquad  \\
 \displaystyle   \displaystyle \frac{2 l}{  D^2}, \qquad\qquad\qquad\qquad\qquad\qquad\qquad\qquad\qquad\qquad\qquad\qquad\qquad\qquad\qquad\qquad\qquad\qquad\qquad\qquad\qquad\qquad\qquad\qquad \ W \leq l \leq D,W \leq z < D\\
 \qquad \qquad\qquad  \\
 \displaystyle
 \frac{l}{D^2}+ \frac{2l}{\pi D^2} \left(\arccos \left(\frac{l}{2W}\right)-\arctan \left(\frac{1}{\sqrt{4 W^2-l^2}} \right) \right)
 , \qquad\qquad\qquad\qquad\qquad\qquad\qquad\qquad\qquad\qquad\qquad\quad 0 \leq l < 2W , 0 \leq z < W \\
 \qquad \qquad\qquad  \\
 \displaystyle 0
, \qquad\qquad\qquad\qquad\qquad\qquad\qquad\qquad\qquad\qquad\qquad\qquad\qquad\qquad\qquad\qquad\qquad\qquad\qquad\qquad\qquad\qquad\qquad\qquad\qquad 2W \leq l \leq D, 0 \leq z < W
\end{array}\right.
\label{eq:kpdf}
\end{equation}
\end{scriptsize}

\begin{scriptsize}
\begin{equation}
F_{L}(l)=\left\{\begin{array}{l}
 \displaystyle   \int_{0}^{l}
 \frac{2 r}{\pi D^{2}} \bigg[\arctan \left(\frac{r}{\sqrt{4 W^2-r^2}}\right)-\arctan\left(\frac{r^2-3 W^2}{\sqrt{10 r^2 W^2 -r^4 -9 W^4}}\right)+4 \arccos \left(\frac{r}{4 W}+\frac{3 W}{4 r}\right)\\
  \displaystyle \qquad\qquad\qquad\qquad\qquad\qquad\qquad\qquad\qquad\qquad\qquad\qquad\qquad\qquad\qquad\qquad\qquad\qquad\qquad\qquad -\arccos\left(\frac{r}{2 W}\right)\bigg] dr , \  0 \leq l < W, W \leq z < D  \\
\qquad \qquad\qquad  \\
 \displaystyle  \frac{l ^ 2}{  D^2}, \qquad\qquad\qquad\qquad\qquad\qquad\qquad\qquad\qquad\qquad\qquad\qquad\qquad\qquad\qquad\qquad\qquad\qquad\qquad\qquad\qquad\qquad\qquad\qquad \ W \leq l \leq D, W \leq z < D \\
 \qquad \qquad\qquad  \\
\displaystyle
\frac{1}{2 \pi  D^2} \bigg[\pi  l^2 -2 l \sqrt{4 W^2-l^2}-2 l^2 \arctan \left(\frac{l}{\sqrt{4 W^2-l^2}}\right)+4 W^2 \arctan \left(\frac{l}{\sqrt{4 W^2-l^2}}\right) +2 l^2 \arccos \left(\frac{l}{2 W}\right)\\
\qquad\qquad\qquad \qquad\qquad\qquad\qquad\qquad\qquad\qquad\qquad\qquad\qquad\qquad\qquad\qquad\qquad\qquad\qquad\quad
 \displaystyle +4 W^2 \arcsin \left(\frac{l}{2 W}\right) \bigg] , \  0 \leq l < 2W,  0 \leq z < W \\
 \displaystyle  \displaystyle  \frac{l ^ 2}{  D^2},\qquad\qquad\qquad\qquad\qquad\qquad\qquad\qquad\qquad\qquad\qquad\qquad\qquad\qquad\qquad\qquad\qquad\qquad\qquad\qquad\qquad\qquad\qquad\qquad \ 2W \leq l \leq D, 0 \leq z < W
\end{array}\right.
\label{eq:kcdf}
\end{equation}
\end{scriptsize}
\setcounter{equation}{20}
\hrulefill
\vspace*{4pt}
\end{figure*}

\subsubsection{RD Model}\label{subsubsec:RD}

Fig.~\ref{fig:Trails-a} presents the trails of an infected individual walking within a circular area with the radius of 100 meters ($\rm{m}$) following the RD model, where the moving time $T$ is 2000 seconds ($\rm{s}$), the pause time is 0.1 $\rm{s}$ and the moving speed varies from 1 meter/second ($\rm{m/s}$) to 5 $\rm{m/s}$. The distance between this infected individual to the susceptible individual at the center of this circular region is a random variable following the uniform distribution.  The PDF of this random variable is the same as that in Eq.~\eqref{eq:fr}.

Therefore, when an infected individual moves following the RD model, the mobility does not impact the infectious probability. This result differs from the results when the infected individual moves following the RWK model and the RWP model, whose moving trails are given in Fig.~\ref{fig:Trails-b} and Fig.~\ref{fig:Trails-c}, respectively. Next, we investigate the effect of mobility on the infectious probability when the infected individuals are moving with the RWK model.

\subsubsection{RWK Model}
We denote the distance between the $i$th infected individual and the reference susceptible individual by $l_{i}$, and the set of $N$ distances is $L_{N}$, and the radius of the circular network region is $D$. For simplicity of analysis, we assume the reference susceptible individual is located at the center of the circular network region. Although many studies also consider the RWK model~\cite{XJunfei:2014ICST,GZhenhua:2011ICC}, both PDF and CDF of distance $l$ are not provided. Therefore, we need to derive PDF and CDF of distance $l$ between an infected individual to the susceptible individual in the RWK model first.

\begin{lemma}
When infected individuals are moving according to the RWK model, PDF $f_{L}(l)$ and CDF $F_{L}(l)$ of the distance between the susceptible individual and the infected individual $l$ is given by \eqref{eq:kpdf} and (\eqref{eq:kcdf}, respectively.
\label{lemma:kpdf}
\end{lemma}

\begin{proof} See Appendix A.

\end{proof}

Based on PDF and CDF of the distance between the susceptible individual and the infected individual $l$, we have the following lemmas.
\begin{lemma}
When infected individuals are moving according to the RWK model, the distance $L_{1}$ between the susceptible individual and the dominant infected individual is given by
\begin{equation}
f_{L_{1}}\left(l_{1}\right)=N\left(1-F_{L}(l_{1})\right)^{N-1} f_{L}(l_{1}),
\label{eq:fL1}
\end{equation}
where $f_{L}(l)$ and $F_{L}(l)$ are given by (\ref{eq:kpdf}) and (\ref{eq:kcdf}), respectively.
\label{lemma:FL1_s}
\end{lemma}

\begin{proof}
Since the dominant infected individual is the closest infected individual to the susceptible individual, the conditional CDF of $L{_1}$ can be computed by the following expression,
\begin{equation}
\begin{aligned}
F_{L_{1}}(l_{1}) &=\mathbb{P}\left(l_{1}\leq l \right) =1-\mathbb{P}\left(\min \left\{L_{N}\right\}> l \right) \\
&=1-\mathbb{P}\left(L_{1}>l, L_{2}>l, \ldots, L_{N}>l \right) \\
& \stackrel{(a)}{=} 1-\left(1-F_{L}(l)\right)^{N},
\end{aligned}
\end{equation}
where $(a)$ follows the i.i.d. nature of the set of distances $L_{N}$. Differentiating the above expression with respect to $l$, PDF of $L{_1}$ is given by
\begin{equation}
\begin{split}
f_{L_{1}}\left(l_{1}\right)&=N\left(1-F_{L}\left(l_{1} \right)\right)^{N-1} f_{L}\left(l_{1} \right).
\label{eq:kpdf1}
\end{split}
\end{equation}
\end{proof}

For the reference susceptible individual situated at the origin, PDF of the distance between the susceptible individual and the minor infected individuals conditioned on the distance $L_{1}$ between the susceptible individual and the dominant infected individual is given by

{\small
\begin{equation}
f_{L_{i}}\left(l_{i} \mid l_{1}\right)=  f_{L}(l_{i}),  \quad l_{1} \leq l_{i} \leq D,
\label{eq:kpdfi}
\end{equation}}

\noindent where $f_{L}(l)$ is given by (\ref{eq:kpdf}).

\begin{lemma}
When the infected individuals are moving according to the RWK model, the mean of the virus volume (excluding the virus volume from the dominant infected individual) conditioned on the distance between the susceptible individual and the dominant infected individual $L_{1}$ is

{\small
\begin{equation}
\begin{aligned}
\mu_{I_{N-1}'} =\frac{ (N-1) (V_{m}+V_{M})}{2} \int_{l_{1}}^{D}    l_{i}^{-\eta}  f_{L}\left(l_{i} \right)   \mathrm{d} l_{i}
\label{eq:kEIN1_m}
\end{aligned}
\end{equation}
}
\label{lemma:kEIN1_m}
\end{lemma}

 \begin{proof} According to the definition of the mean value of a random variable, we have the expression of the mean of virus volume from minor infected individuals as

{\small
\begin{equation}
\begin{aligned}
\mu_{I_{N-1}'} =& \mathbb{E}\left[I_{N-1} \mid L_{1}, V_{1} \right]\\
\stackrel{(a)}{=}&(N-1) \mathbb{E}\left[V_{i} \cdot L_{i}^{-\eta} \mid L_{1}, V_{1}\right]\\
= &(N-1) \int_{V_{m}}^{V_{M}}   \int_{l_{1}}^{D}  V_{i} \cdot  l_{i}^{-\eta}   f_{L_{i}}  \left(l_{i} \mid l_{1}\right) {f_V}(V_{i}) \mathrm{d} l_{i}  \mathrm{d} V_{i}\\
\stackrel{(b)}{=}&(N-1) \int_{V_{m}}^{V_{M}} V_{i} {f_V}(V_{i}) \mathrm{d} V_{i} \cdot \int_{l_{1}}^{D}    l_{i}^{-\eta} f_{L_{i}}  \left(l_{i} \mid l_{1}\right) \mathrm{d} l_{i},
\label{eq:murwp}
\end{aligned}
\end{equation}}

\noindent where $(a)$ follows  the conditional i.i.d. nature of the distance \{$L_{i}$\} and the i.i.d. nature of the virus volume \{$V_{i}$\}, $(b)$ follows the independence between \{$L_{i}$\} and \{$V_{i}$\}.
Substituting (\ref{eq:fv}) and (\ref{eq:kpdfi}) into \eqref{eq:murwp}, we have

{\small
\begin{equation}
\begin{aligned}
\mu_{I_{N-1}'} = \frac{(N-1) (V_{m}+V_{M})}{2} \int_{l_{1}}^{D}    l_{i}^{-\eta} f_{L}\left(l_{i} \right)   \mathrm{d} l_{i}.
\end{aligned}
\end{equation}
}
 \end{proof}

\begin{lemma}

When the infected individuals are moving according to the RWK model, the variance of virus volume (excluding the virus volume from the dominant infected individual) conditioned on the distance between the susceptible individual and the dominant infected individual $L_{1}$ is

{\small
\begin{equation}
\begin{aligned}
{\sigma_{I_{N-1}'}}&=(N-1) \bigg[ \frac{ (V_{m}^2+V_{m}V_{M}+ V_{M}^2) }{3} \int_{r_{1}}^{D} l_{i}^{-2 \eta} f_{L}(l_{i}) \mathrm{d} r_{i} \\
&\qquad\qquad -  \frac{(V_{m}+ V_{M})^2 } {4} \left(\int_{r_{1}}^{D}l_{i}^{- \eta} f_{L}(l_{i}) \mathrm{d} r_{i} \right)^2 \bigg].
\label{eq:ksigma_m}
\end{aligned}
\end{equation}
}

\noindent where $f_{L}(l)$ is given by (\ref{eq:kpdf}).
\label{lemma:kVIN1_m}
\end{lemma}

\begin{proof}
The conditional variance of virus volume at the susceptible individual  excluding the virus volume from the dominant infected individual can be calculated by

{\small
\begin{equation}
\begin{aligned}
{\sigma_{I_{N-1}'}}&=\operatorname{Var}\left[I_{N-1} \mid L_{1}, V_{1}\right] \\
&=(N-1) \bigg[\int_{V_{m}}^{V_{M}}  \int_{l_{1}}^{D}  {(V_{i} l_{i}^{- \eta})}^2 {f_V}(V_{i}) f_{L_{i}}\left(l_{i} \mid  l_{1}\right) \mathrm{d} l_{i}\mathrm{d} V_{i}  \\
&\quad\quad\quad-\left( \int_{V_{m}}^{V_{M}}   \int_{l_{1}}^{D}  V_{i} \cdot  l_{i}^{-\eta} {f_V}(V_{i})  f_{L_{i}}  \left(l_{i} \mid l_{1}\right) \mathrm{d} l_{i}  \mathrm{d} V_{i}  \right)^{2}\bigg].
\label{eq:sigma}
\end{aligned}
\end{equation}
}

After substituting (\ref{eq:fv}) and (\ref{eq:kpdfi}) into \eqref{eq:sigma}, we derive

{\small
\begin{equation}
\begin{aligned}
{\sigma_{I_{N-1}'}}=&(N-1) \bigg[     \int_{V_{m}}^{V_{M}}  \frac{ {V_{i}}^2 } {V_{M}-V_{m}} \mathrm{d} V_{i}   \int_{r_{1}}^{D} l_{i}^{-2 \eta} f_{L}(l_{i})   \mathrm{d} r_{i}\\
& \qquad\qquad  - \bigg( \int_{V_{m}}^{V_{M}}  \frac{ {V_{i}} } {V_{M}-V_{m}} \mathrm{d} V_{i}  \int_{r_{1}}^{D} l_{i}^{- \eta} f_{L}(l_{i}) \mathrm{d} r_{i} \bigg)^2  \bigg]\\
= &(N-1) \bigg[ \frac{ (V_{m}^2+V_{m}V_{M}+ V_{M}^2) }{3} \int_{r_{1}}^{D} l_{i}^{-2 \eta} f_{L}(l_{i}) \mathrm{d} r_{i} \\
&\qquad\qquad -  \frac{(V_{m}+ V_{M})^2 } {4} \left(\int_{r_{1}}^{D}l_{i}^{- \eta} f_{L}(l_{i}) \mathrm{d} r_{i} \right)^2 \bigg].
\end{aligned}
\end{equation}
}

\end{proof}

Based on the above Lemmas, we next derive the infectious probability as follows.
\begin{theorem}
\label{theorem:kPinf_m}
When the infected individuals are moving according to the RWK model, the infectious probability of a susceptible individual is expressed by

\begin{equation}
\begin{aligned}
\mathbb{P}_\text{inf} \approx   \int_{V_{m}}^{V_{M}} \int_{0}^{D}Q \left(\frac{V_\text{th} - V_{1}\cdot l_{1}^{-\eta}-\mu_{I'_{N-1}}}{\sigma_{I'_{N-1}}}\right)\\
 \times \frac{N\left(1-F_{L}(l_{1})\right)^{N-1} f_{L}(l_{1})}{ (V_{M}-V_{m})} \mathrm{d} l_{1} \mathrm{d} V_{1},
\label{eq:kPinfF}
\end{aligned}
\end{equation}
where $Q(\cdot)$ is the $Q$ function with the expression $Q(x)=\int_{x}^{+\infty} \frac{1}{\sqrt{2 \pi}} \exp \left(-\frac{1}{2} t^{2}\right) dt$, $f_{L}(l)$  is given by (\ref{eq:kpdf}) and $F_{L}(l)$ is given by (\ref{eq:kcdf}).
\end{theorem}


\begin{proof}By definition,  the infectious probability of a susceptible individual when  the infected individuals are moving according to the RWK model is given by
\begin{equation}
\begin{aligned}
\mathbb{P}_\text{inf} =  \int_{V_{m}}^{V_{M}} \int_{0}^{D} P\left(I_\text{inf}>V_\text{th} \mid  L_{1}, V_{1}\right) & f_{L_{1}}( l_{1})  {f_V}(V_{1}) \mathrm{d} l_{1}  \mathrm{d} V_{1}.
\label{eq:kPinf1}
\end{aligned}
\end{equation}

As stated earlier, $I_{N-1}$ is the sum of i.i.d. random variables.
Therefore, by applying CLT, the infectious probability can be computed by
CDF of a Gaussian random variable. We then have
\begin{equation}
\begin{aligned}
\label{eq:kPinf2}
P\left(I_\text{inf}>V_\text{th} \mid  L_{1}, V_{1} \right)
=& P\left(\sum_{i \in \Phi } V_{i} \cdot l_{i}^{-\eta} >V_\text{th} \mid  L_{1}, V_{1}\right) \\
=& P\left( I_{N-1} >V_\text{th} - V_{1} \cdot l_{1}^{-\eta} \right) \\
\approx & Q\left(\frac{V_\text{th}- V_{1} \cdot l_{1}^{-\eta}-\mu_{I'_{N-1}}}{\sigma_{I'_{N-1}}}\right).
\end{aligned}
\end{equation}
Substituting (\ref{eq:fv}), (\ref{eq:kpdf1}) and (\ref{eq:kPinf2}) into (\ref{eq:kPinf1}), we derive the expression of $\mathbb{P}_\text{inf}$ in (\ref{eq:kPinfF}).
\end{proof}

\subsubsection{RWP Model}
We denote the distance between the $i$th infected individual and the reference susceptible individual by $u_{i}$, and the set of $N$ distances is $U_{N}$, and the radius of the circular network region is $D$. For simplicity of analysis, we assume the reference susceptible individual is located at the center of the circular network region. There are a number of studies on the RWP model~\cite{Gong:2014TMC, Irio:2018TCOM} and PDF of the distance $u$ between an infected individual and the susceptible individual is given by~\cite{Hyytia:2006TMC} as follows,

\begin{equation}
\begin{aligned}
f_{U}(u)=& \frac{2\pi u h(u)}{2 \pi \int_{0}^{1} u \cdot h(u) d u}\\
=&\frac{45 u \left(1-u^{2}\right)}{64} \int_{0}^{\pi} \sqrt{1-u^{2} \cos ^{2} \phi} d \phi \\
\stackrel{(a)}{\approx} &\sum_{i=1}^{3} \frac{B_{i}}{D^{\beta_{i}+1}} u^{\beta_{i}}, 0 < u \le 1 ,
\end{aligned}
\end{equation}
where $h(u)=2\left(1-u^{2}\right) \int_{0}^{\pi} \sqrt{1-u^{2} \cos ^{2} \phi} d \phi$, $B_{i}=(1/73) \cdot [324,-420,96] $, $\beta_{i}=[1,3,5]$, and $(a)$ is derived because applying the exact pdf in complicated problems may be time consuming as the exact expression of $f(u)$ requires the numerical integration. This simplified PDF expression of $f(u)$, which was proved to be accurate, is provided by~\cite{Fernandez:2018TVT}.

\begin{lemma}
When infected individuals are moving according to the RWP model, the distance between the susceptible individual and the dominant infected individual $R_{1}$ is given by
\begin{equation}
f_{U_{1}}\left(u_{1}\right) = N \Bigg(\Big(1-\sum_{j=1}^{3}   \frac{C_{i}}{D^{\beta_{j}+1}} u^{\gamma_{j}}\Big)^{N-1}\Bigg) \left(\sum_{i=1}^{3} \frac{B_{i}}{D^{\beta_{i}+1}} u^{\beta_{i}} \right).
\label{eq:fU1}
\end{equation}
\label{lemma:FR1_m}
\end{lemma}

\begin{proof}
Similar to Lemma \ref{lemma:FR1_s}, we have the conditional CDF of $u_{1}$ by

\begin{equation}
\begin{aligned}
F_{U_{1}}(u_{1}) &=\mathbb{P}\left(u_{1}\leq u \right) =1-\mathbb{P}\left(\min \left\{U_{N}\right\}> u \right) \\
&=1-\mathbb{P}\left(U_{1}>u, U_{2}>u, \ldots, U_{N}>u \right) \\
& = 1-\left(1-F_{U}(u)\right)^{N}.
\end{aligned}
\end{equation}

Differentiating the above expression with respect to $u$, PDF of $u_{1}$
is obtained as
\begin{equation}
\begin{aligned}
f_{U_{1}}\left(u_{1}\right) = & N\left(1-F_{U}\left(u \right)\right)^{N-1} f_{U}\left(u \right)\\
= & N \Bigg(\Big(1-\sum_{j=1}^{3} \frac{C_{i}}{D^{\beta_{j}+1}} u^{\gamma_{j}}\Big)^{N-1}\Bigg) \left(\sum_{i=1}^{3} B_{i} u^{\beta_{i}} \right),
\label{eq:pfU1}
\end{aligned}
\end{equation}
where $B_{i}=(1/73) \cdot [324,-420,96] $, $\beta_{i}=[1,3,5]$,  $C_{j}=(1/73) \cdot [162,-105,16] $ and $\gamma_{j}=[2,4,6]$ as given by \cite{Fernandez:2018TVT}.
\end{proof}

For the reference susceptible individual situated at the origin, PDF of the distances between the susceptible individual and the minor infected individuals conditioned on the distance between the susceptible individual and the dominant infected individual $U_{1}$ is given by
\begin{equation}
f_{U_{i}}\left(u_{i} \mid u_{1}\right)= \frac{ u_{i} h(u_{i})}{ \int_{u_{1}}^{1} u_{i} \cdot h(u_{i}) d u_{i}},  u_{1} \leq u_{i} \leq 1,
\end{equation}
where $h(u_{i})=2\left(1-u_{i}^{2}\right) \int_{0}^{\pi} \sqrt{1-u_{i}^{2} \cos ^{2} \phi} d \phi$.

\begin{lemma}
When the infected individuals are moving according to the RWP model, the mean of virus volume (excluding the virus volume from the dominant infected individual) conditioned on the distance between the susceptible individual and the dominant infected individual $U_{1}$ is
\begin{equation}
\begin{aligned}
\mu_{I_{N-1}''} =\frac{(N-1) (V_{m}+V_{M}) \int_{u_{1}}^{1}  u_{i}^{1 -\eta} h(u_{i})\mathrm{d} u_{i}}{ 2 \int_{u_{1}}^{1} u_{i} \cdot h(u_{i}) d u_{i}},
\label{eq:EIN1_m}
\end{aligned}
\end{equation}
\label{lemma:EIN1_m}
\end{lemma}

 \begin{proof} With the similar approach given in Lemma \ref{lemma:EIN1_s}, according to the definition of the mean of a variable, we have the expression of the mean of virus volume from minor infected individuals as

\begin{equation}
\begin{aligned}
\mu_{I_{N-1}''} =& \mathbb{E}\left[I_{N-1}'' \mid U_{1}, V_{1}\right] \\
=&(N-1) \mathbb{E}\left[V_{i} \cdot U_{i}^{-\eta} \mid U_{1}, V_{1}\right]\\
= &(N-1) \int_{V_{m}}^{V_{M}}   \int_{r_{1}}^{1}  V_{i} \cdot  r_{i}^{-\eta}   f_{R_{i}}  \left(r_{i} \mid r_{1}\right) {f_V}(V_{i}) \mathrm{d} r_{i}  \mathrm{d} V_{i}\\
=&\frac{(N-1) (V_{m}+V_{M}) \int_{u_{1}}^{1}  u_{i}^{1 -\eta} h(u_{i})\mathrm{d} u_{i}}{ 2 \int_{u_{1}}^{1} u_{i} \cdot h(u_{i}) d u_{i}}.
\end{aligned}
\end{equation}
 \end{proof}

\begin{lemma}

When the infected individuals are static, the variance of virus volume (excluding virus volume from the dominant infected individual) conditioned on the distance between the susceptible individual and the dominant infected individual $U_{1}$ is

{\small
\begin{equation}
\begin{aligned}
{\sigma_{I_{N-1}''}}&=(N-1)\left[ \frac{(V_{m}^2+V_{m}V_{M}+ V_{M}^2) }{3} \int_{u_{1}}^{1} \frac{u_{i}^{1-2\eta} h(u_{i})}{2 \pi \int_{u_{1}}^{1} u \cdot h(u_{i}) d u_{i}} \right.\\
&\quad\quad\quad\quad\left.\quad- \frac{ (V_{m}+V_{M})^2}{4} \left(\int_{u_{1}}^{1} \frac{u_{i}^{1 -\eta} h(u_{i})}{2 \pi \int_{u_{1}}^{1} u \cdot h(u_{i}) d u_{i}}\right)^2\right].
\label{eq:VIN1_m}
\end{aligned}
\end{equation}
}
\label{lemma:VIN1_m}
\end{lemma}

\begin{proof}
With the similar approach given in Lemma \ref{lemma:VIN1_s}, the conditional variance of virus volume at the susceptible individual excluding the virus volume from the dominant infected individual can be calculated by

{\footnotesize
\begin{equation}
\begin{aligned}
{\sigma_{I_{N-1}''}}&=\operatorname{Var}\left[I_{N-1}' \mid U_{1}, V_{1}\right] \\
&=(N-1) \bigg[\int_{V_{m}}^{V_{M}}  \int_{u_{1}}^{D}  {(V_{i} u_{i}^{- \eta})}^2 {f_V}(V_{i}) f_{U_{i}}\left(u_{i} \mid  u_{1}\right) \mathrm{d} u_{i}\mathrm{d} V_{i}  \\
&\quad\quad\quad-\left( \int_{V_{m}}^{V_{M}}   \int_{u_{1}}^{D}  V_{i} \cdot  u_{i}^{-\eta} {f_V}(V_{i})  f_{U_{i}}  \left(u_{i} \mid u_{1}\right) \mathrm{d} u_{i}  \mathrm{d} V_{i}  \right)^{2} \bigg]\\
&=(N-1)\left[ \frac{(V_{m}^2+V_{m}V_{M}+ V_{M}^2) }{3} \int_{u_{1}}^{1} \frac{u_{i}^{1-2\eta} h(u_{i})}{2 \pi \int_{u_{1}}^{1} u \cdot h(u_{i}) d u_{i}} \right.\\
&\quad\quad\quad\quad\left.\quad- \frac{ (V_{m}+V_{M})^2}{4} \left(\int_{u_{1}}^{1} \frac{u_{i}^{1 -\eta} h(u_{i})}{2 \pi \int_{u_{1}}^{1} u \cdot h(u_{i}) d u_{i}}\right)^2\right].
\end{aligned}
\end{equation}
}
\end{proof}

We then have the following theorem of the infectious probability with the RWP model.
\begin{theorem}
\label{theorem:Pinf_m}
When the infected individuals are moving according to the RWP model, the infectious probability of a susceptible individual is expressed by
\begin{equation}
\begin{aligned}
\mathbb{P}_\text{inf} \approx \int_{V_{m}}^{V_{M}} \int_{0}^{1} & \left[Q\left(\frac{V_\text{th} - V_{1}\cdot u_{1}^{-\alpha}-\mu_{I_{N-1}''}}{\sigma_{I_{N-1}''}}\right)\right] \\
&\qquad\qquad\quad \times \frac{f_{U_{1}}( u_{1})}{V_{M}-V_{m} } \mathrm{d} u_{1}  \mathrm{d} V_{1},
\end{aligned}
\end{equation}
\end{theorem}
where $f_{U_{1}}( u_{1})$ is given by (\ref{eq:pfU1}).

\begin{figure*}[t]
\centering
\subfigure[$ N = 20 $, $ D=20$m]{
\label{fig:Pinf_V-a}
\includegraphics[width=5.8cm]{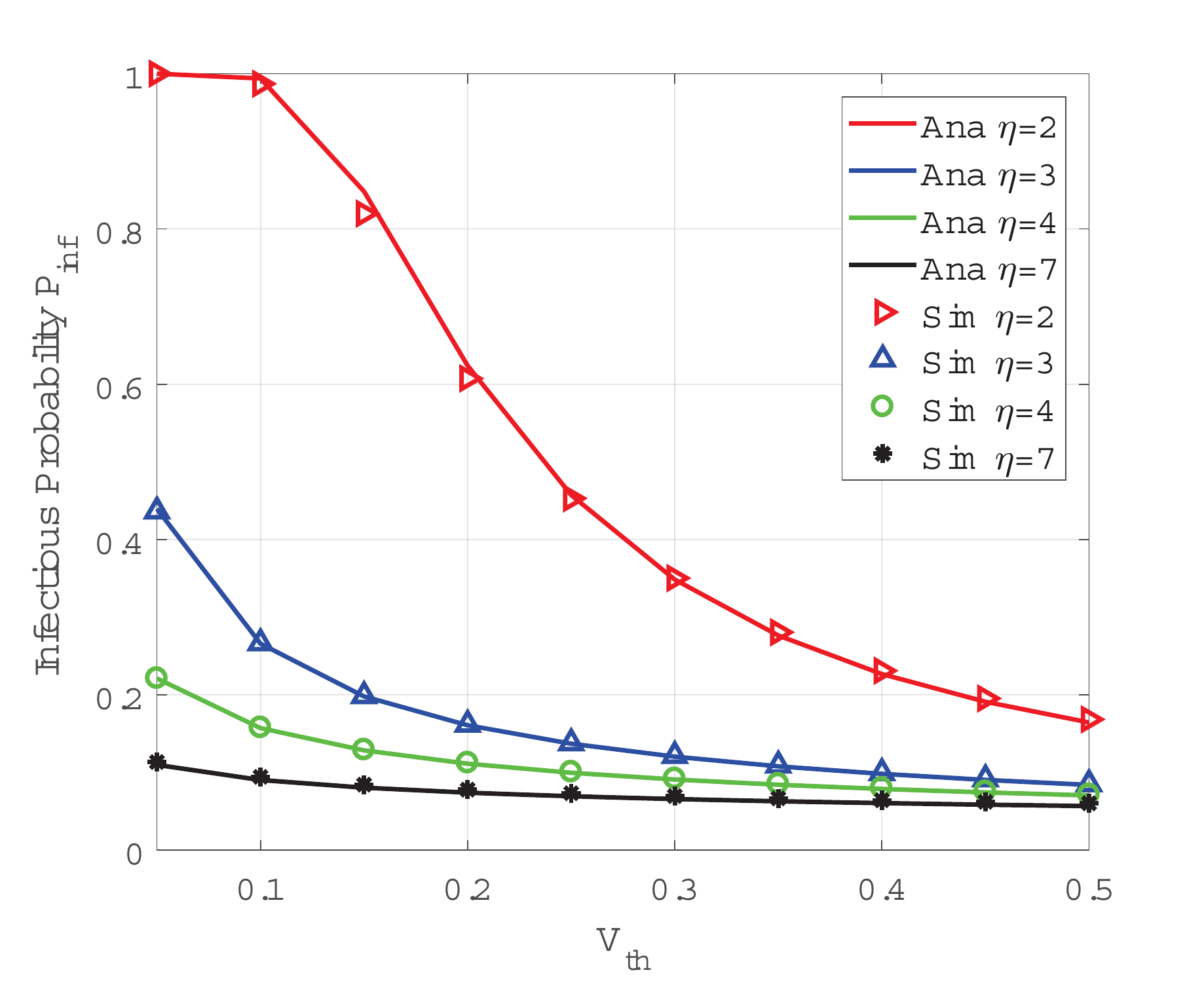}}\hfil
\subfigure[$ N = 20$, $ D=100$m]{
\label{fig:Pinf_V-b}
\includegraphics[width=5.8cm]{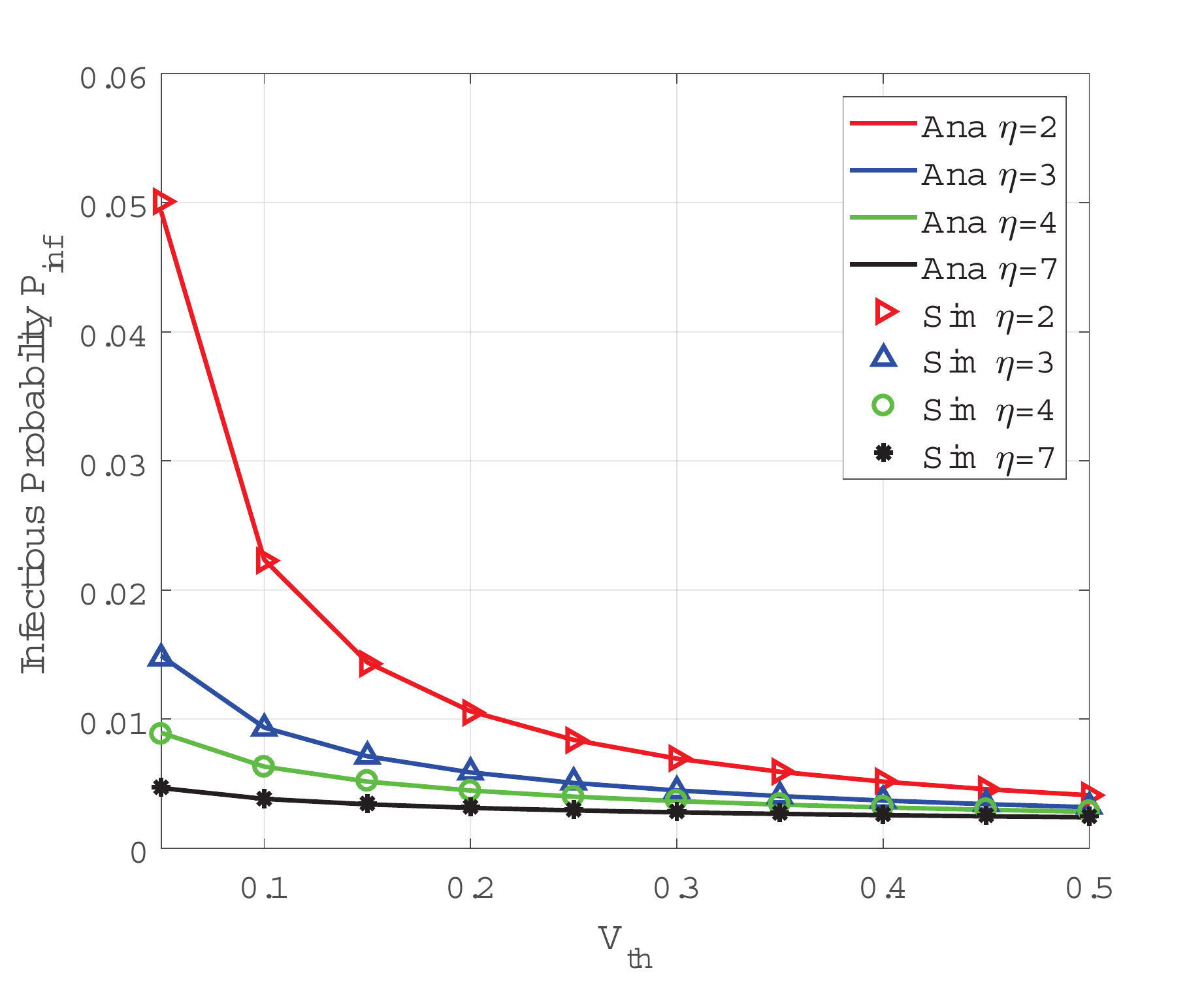}}\hfil
\subfigure[$N = 10$, $ D=100$m]{
\label{fig:Pinf_V-c}
\includegraphics[width=5.8cm]{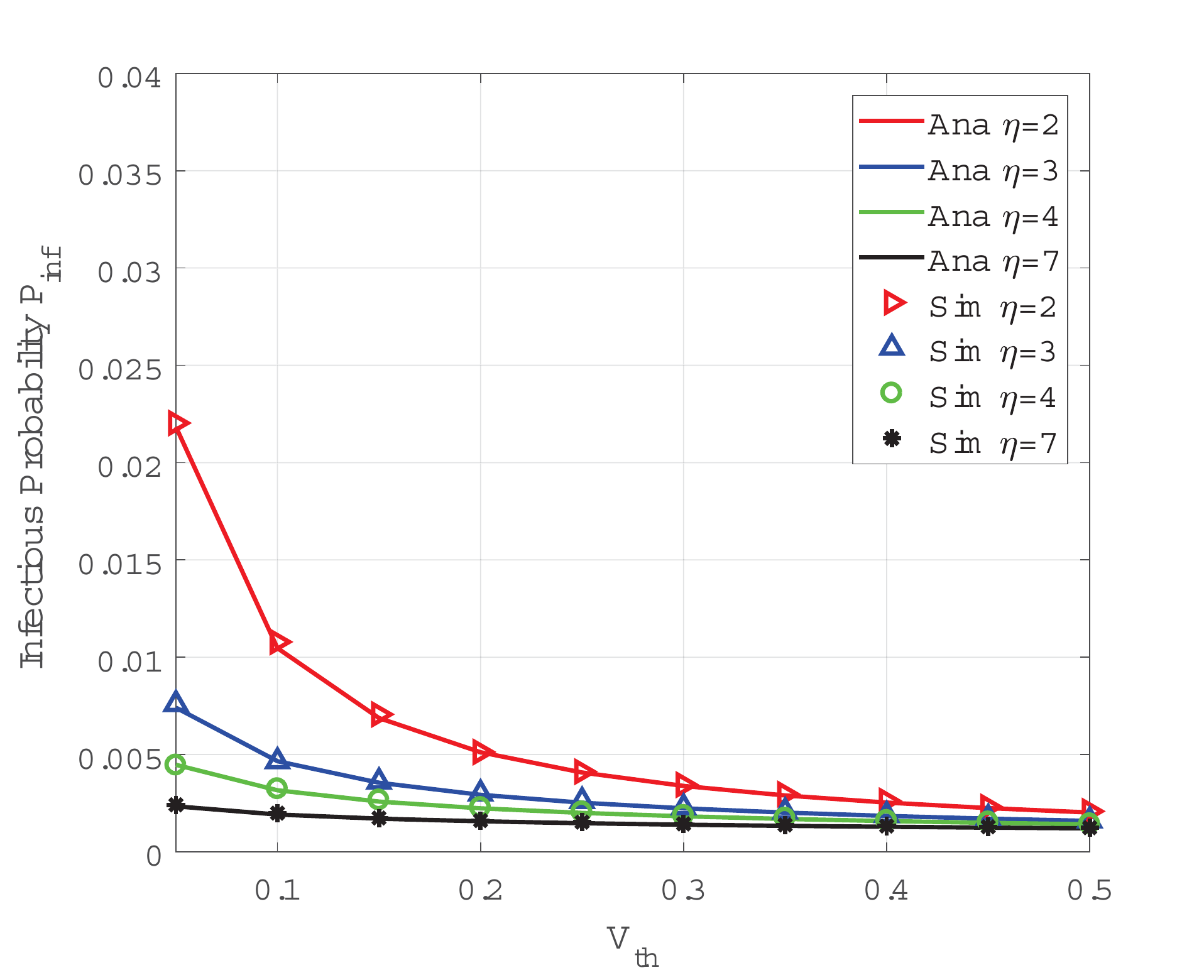}}\hfil
\caption{Infectious probability $\mathbb{P}_\text{inf}$ versus the threshold virus volume ${V}_\text{th}$ in a small network area when infected individuals are static.}
\label{fig:Pinf_V}
\end{figure*}

\begin{proof}By definition,  the infectious probability of a susceptible individual when  the infected individuals are moving according to the RWP model is given by

{\small
\begin{equation}
\begin{aligned}
\mathbb{P}_\text{inf} =  \int_{V_{m}}^{V_{M}} \int_{0}^{1} P\left(I_\text{inf}>V_\text{th} \mid  U_{1}, V_{1}\right) & f_{U_{1}}( u_{1})  {f_V}(V_{1}) \mathrm{d} u_{1}  \mathrm{d} V_{1}\\
\int_{V_{m}}^{V_{M}} \int_{0}^{1} P\left(I_\text{inf}>V_\text{th} \mid  U_{1}, V_{1}\right) &
\cdot \frac{f_{U_{1}}( u_{1})}{V_{M}-V_{m} }  \mathrm{d} u_{1}  \mathrm{d} V_{1}.
\end{aligned}
\end{equation}
}

As stated earlier, $I_{N-1}$ is the sum of i.i.d. random variables.
Therefore, by applying CLT, the infectious probability can be computed by
CDF of a Gaussian random variable. Therefore, we have
\begin{equation}
\begin{aligned}
P\left(I_\text{inf}>V_\text{th} \mid  U_{1}, V_{1}\right)
=& P\left(\sum_{i \in \Phi } V \cdot u_{i}^{-\eta} >V_\text{th} \mid  U_{1}, V_{1}\right) \\
=& P\left( I_{N-1} >V_\text{th} -V_{1}\cdot u_{1}^{-\eta} \right) \\
=& Q\left(\frac{V_\text{th}-V_{1}\cdot u_{1}^{-\eta}-\mu_{I_{N-1}''}}{\sigma_{I_{N-1}''}}\right).
\end{aligned}
\end{equation}
\end{proof}

\section{Numerical Results}
\label{sec:result}

In this section, we conduct extensive Monte-Carlo simulations to compare simulation results with analytical results to evaluate the accuracy of the proposed analytical models. Throughout the entire section, the curves represent the analytical results and the markers represent the simulation results. We present the numerical results of the infectious probability $\mathbb{P}_\text{inf}$ with consideration of both static infected individuals and moving infected individuals (with different moving models).

\subsection{Infectious Probability with Static Individuals}
\label{sec:resultA}
We first present the numerical results when the infected individuals are static.  In the following results, the virus volume is normalized and the minimal virus volume $V_{m}$ is chosen to be 0.5, and the maximal virus volume $V_{M}$ is set to be 1.5. The path loss factor $\eta$ is chosen from the range $[2,7]$~\cite{LIU2021106542}.

\begin{figure}[h]
\centering
\includegraphics[width=7.8cm]{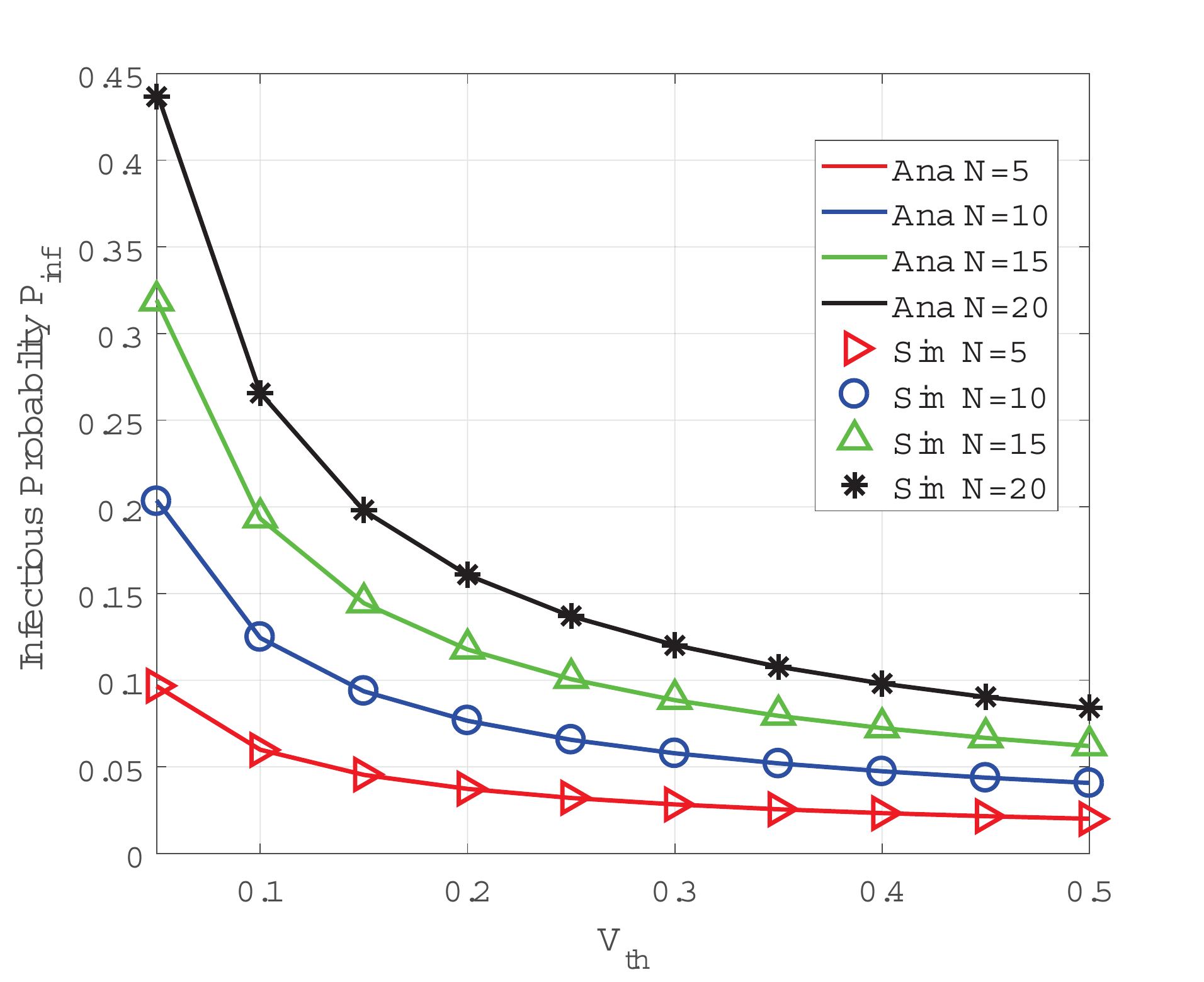}
\caption{Infectious probability $\mathbb{P}_\text{inf}$ versus the threshold of virus volume with $\eta=3$ and $ D=20$m when infected individuals are static.}
\label{fig:Pinf_N}
\end{figure}

Fig.~\ref{fig:Pinf_V} presents the first set of results of infectious probability $\mathbb{P}_\text{inf}$ in a small network area. In this set of results, we compare the infectious probability $\mathbb{P}_\text{inf}$ with the different number of infected individuals $N$, the radius of the virus spreading region $D$ and path loss factor $\eta$. From Fig.~\ref{fig:Pinf_V-a}, we find that $\mathbb{P}_\text{inf}$ significantly decreases when the threshold virus volume $V_\text{thr}$ increases. In this figure, different colors of the curves and markers represent the results with different values of path loss factor $\eta$. When path loss factor $\eta$ decreases, $\mathbb{P}_\text{inf}$ increases quickly, especially in Fig.~\ref{fig:Pinf_V-a}. The number of infected individuals $N$ is 20 in both Fig.~\ref{fig:Pinf_V-a} and Fig.~\ref{fig:Pinf_V-b}, while the radius of the results in Fig.~\ref{fig:Pinf_V-a} is smaller than that in Fig.~\ref{fig:Pinf_V-b}; it implies the higher density of infected individuals in Fig.~\ref{fig:Pinf_V-a} than that in Fig.~\ref{fig:Pinf_V-b}. As a result, $\mathbb{P}_\text{inf}$ in Fig.~\ref{fig:Pinf_V-a} is much higher than that in Fig.~\ref{fig:Pinf_V-b}. For example, when the virus volume threshold $V_\text{th}$ is 0.1 and path loss factor $\eta$ is 2, $\mathbb{P}_\text{inf}$ in Fig.~\ref{fig:Pinf_V-a} is 0.9904 while $\mathbb{P}_\text{inf}$ in Fig.~\ref{fig:Pinf_V-b} is only 0.2241. When the radius of the virus spreading region $D$ remains to be 100m, and the number of infected individuals reduces from 20 to 10 as shown in Fig.~\ref{fig:Pinf_V-b} and Fig.~\ref{fig:Pinf_V-c}, respectively, we find $\mathbb{P}_\text{inf}$ significantly decreases. For example, when the virus volume threshold $V_\text{th}$ is 0.1, and path loss factor $\eta$ is 2, $\mathbb{P}_\text{inf}$ in Fig.~\ref{fig:Pinf_V-b} is 0.2241, while $\mathbb{P}_\text{inf}$ in Fig.~\ref{fig:Pinf_V-c} is only 0.0113.

\begin{figure*}[t]
\centering
\subfigure[$ N = 1000 $, $ D=2000$m]{
\includegraphics[width=5.8cm]{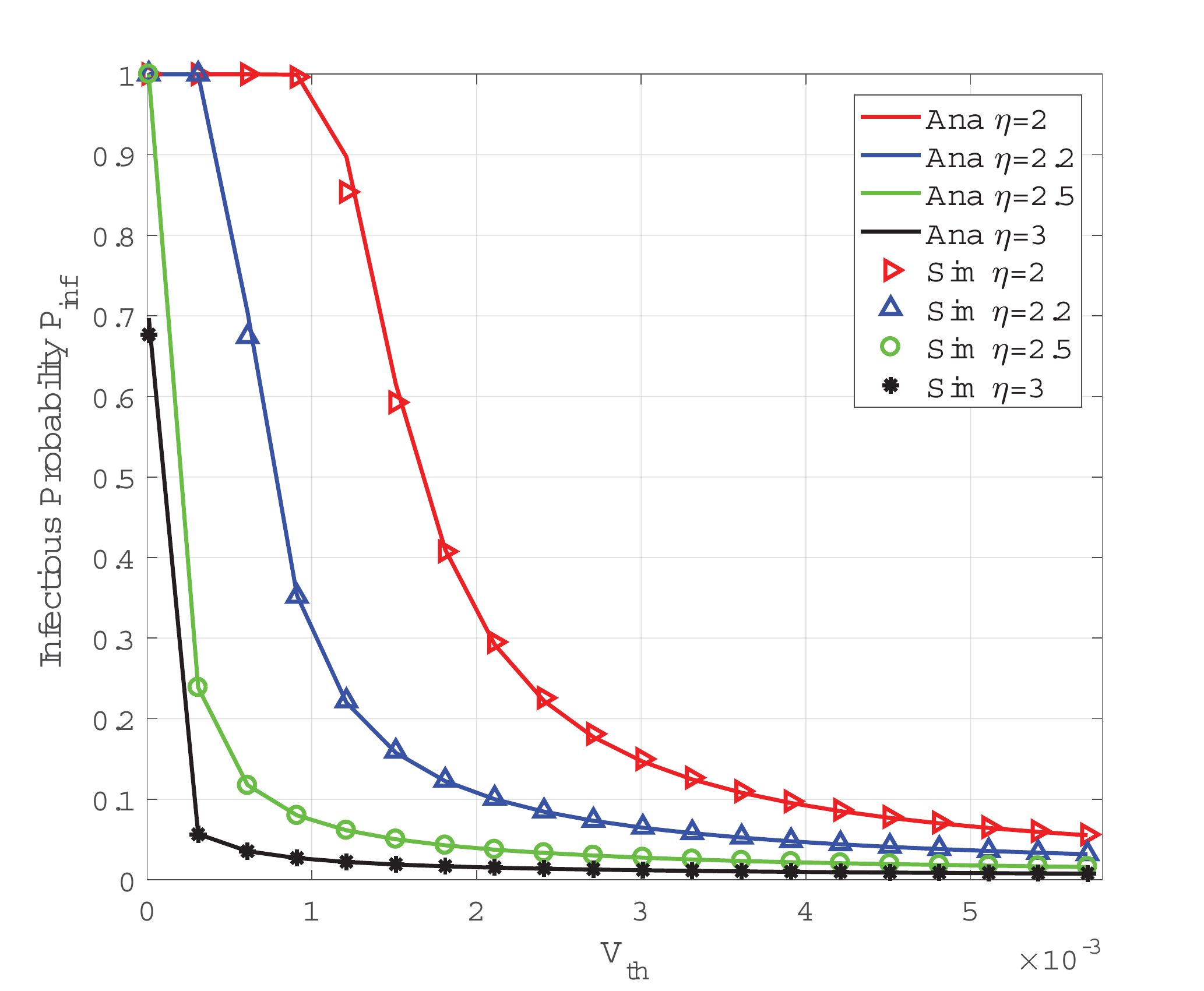}}\hfil
\subfigure[$ N = 2000$, $ D=2000$m]{
\includegraphics[width=5.8cm]{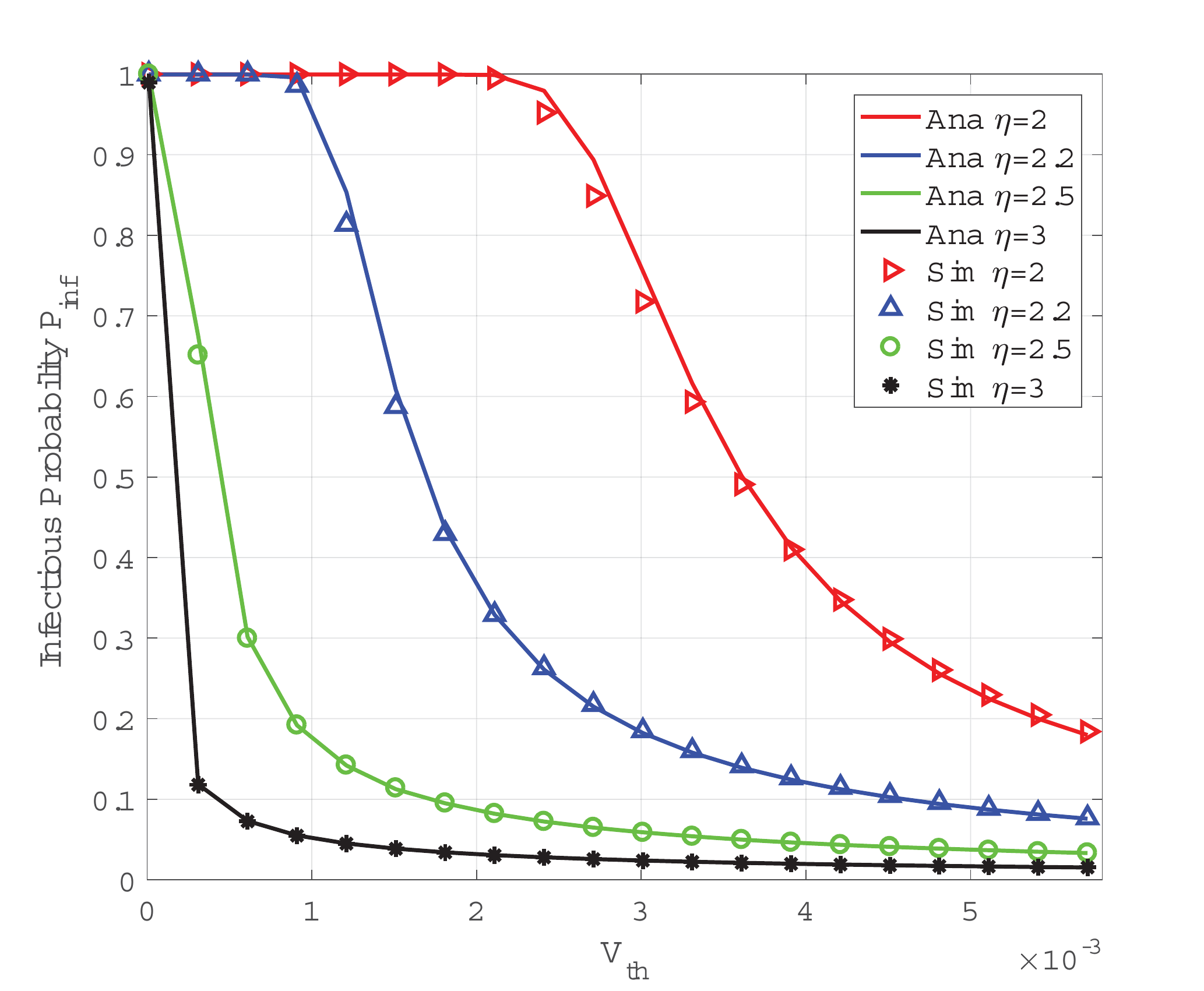}}\hfil
\subfigure[$N = 2000$, $ D=5000$m]{
\includegraphics[width=5.8cm]{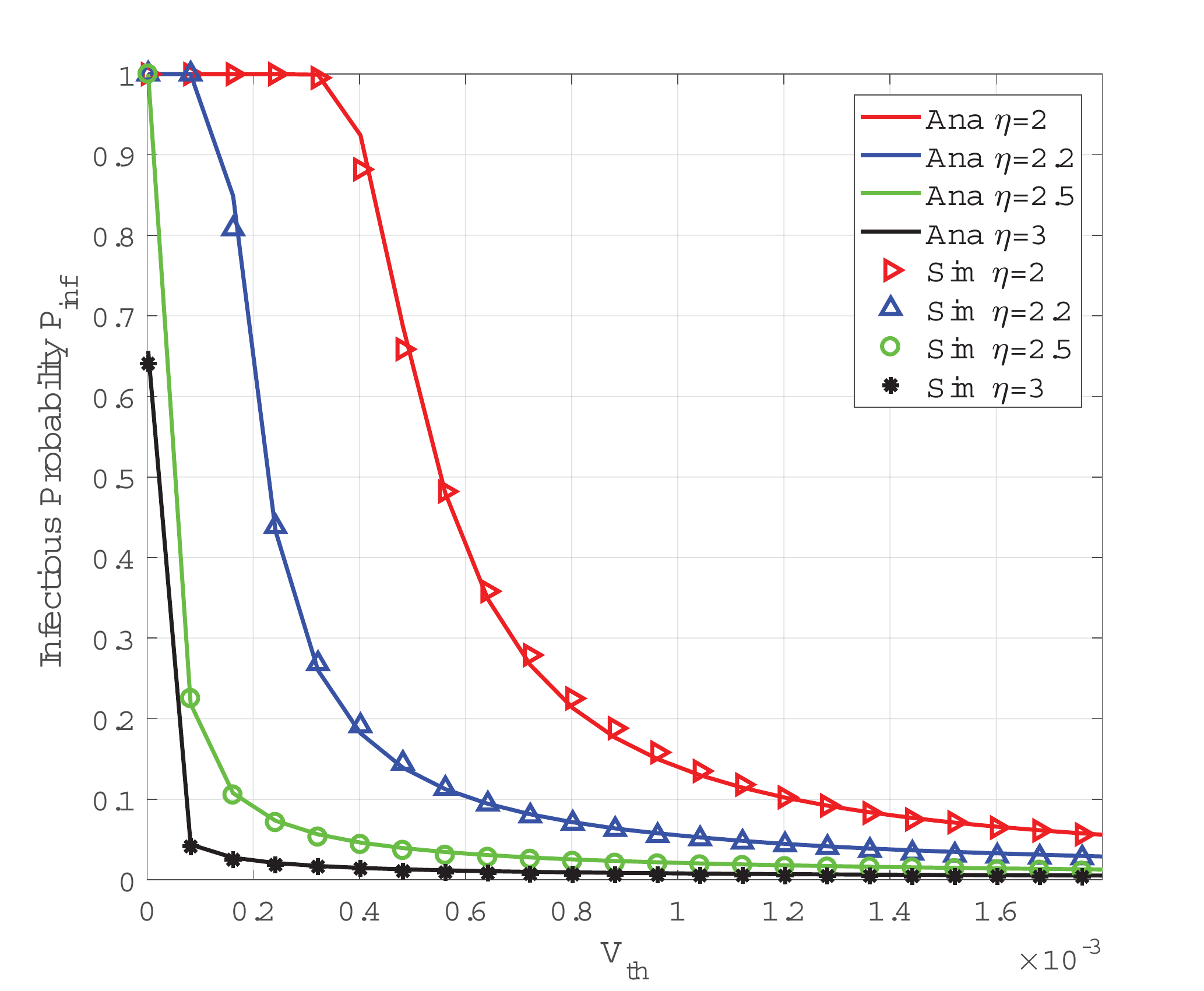}}\hfil
\caption{Infectious probability $\mathbb{P}_\text{inf}$ versus the threshold virus volume $V_\text{th}$ in large network area when infected individuals are static.}
\label{fig:Pinf_V_large}
\end{figure*}


These results in Fig.~\ref{fig:Pinf_V} suggest two potentially effective methods to mitigate $\mathbb{P}_\text{inf}$:  1) increasing the path loss factor $\eta$ of virus spreading and 2) reducing the density of infected individuals around the reference susceptible individual. These methods confirm with previous studies~\cite{ABBOAHOFFEI2021100013,SUN2020102390} that effective countermeasures against the spreading of COVID-19 disease include keeping social distance, avoiding overcrowded environment, and wearing a surgical face mask.


\begin{figure}[ht]
\centering
\includegraphics[width=7.5cm]{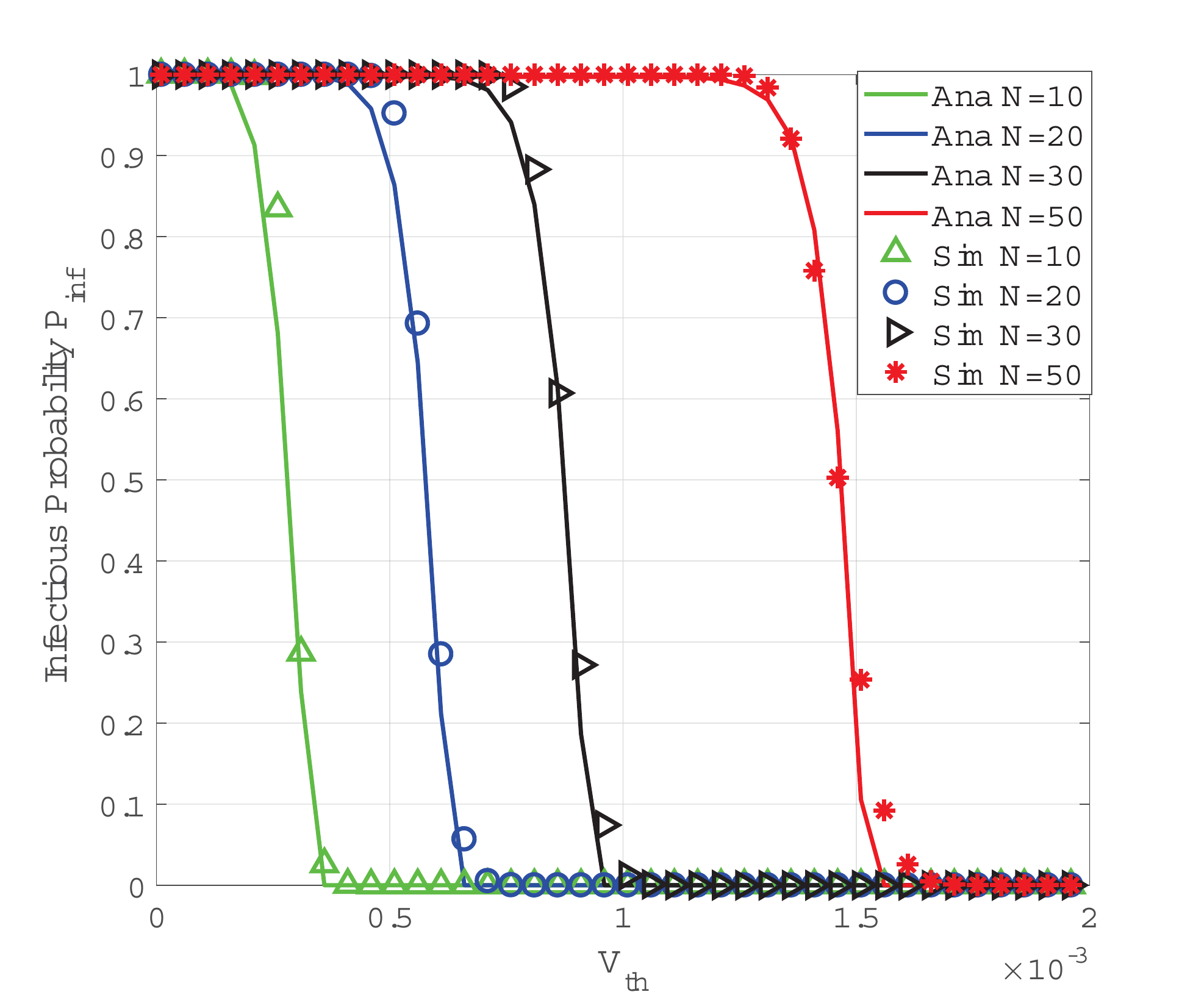}
\caption{Infectious probability $\mathbb{P}_\text{inf}$ versus the threshold virus volume ${V}_\text{th}$ with path loss factor $\eta = 2.5$, network radius $D=100$m and moving distance $W=20$m when infected individuals are moving with the RWK model.}
\label{fig:RWK_N_D100}
\end{figure}

\begin{figure}[ht]
\centering
\includegraphics[width=7.2cm]{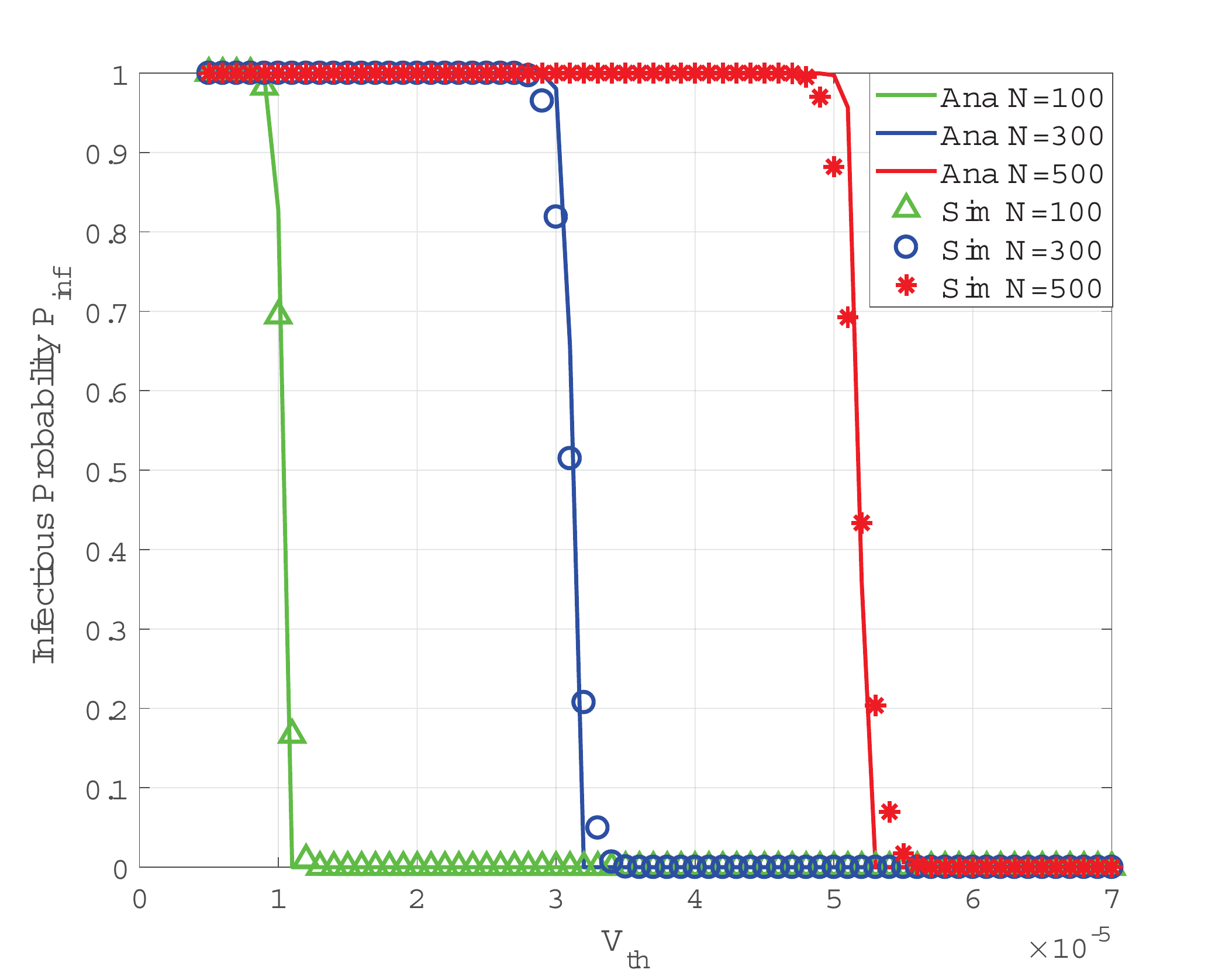}
\caption{Infectious probability $\mathbb{P}_\text{inf}$ versus the threshold virus volume ${V}_\text{th}$ with path loss factor $\eta = 2.5$, network radius $D=1000$m and moving distance $W=200$m when infected individuals are moving with the RWK model.}
\label{fig:RWK_N_D1000}
\end{figure}

Fig.~\ref{fig:Pinf_N} presents the results of infectious probability $\mathbb{P}_\text{inf}$ versus the threshold of virus volume when $\eta=3$ and $ D=20$. In this figure, different colors of the curves and markers represent different numbers of the infected individuals. From this figure, we find the increased number of infected individuals $N$ will lead to the increase of infectious probability $\mathbb{P}_\text{inf}$. For example, when virus volume threshold $V_\text{th}$ is 0.1 and $N$ is 2, $\mathbb{P}_\text{inf}$ in Fig.~\ref{fig:Pinf_V-b} is 0.064, while $\mathbb{P}_\text{inf}$ is increased to 0.1962 when $N$ becomes 15.

Fig.~\ref{fig:Pinf_V_large} presents the set of results of infectious probability $\mathbb{P}_\text{inf}$ in a large network area with a large number of infected individuals. In this set of results, we find that $\mathbb{P}_\text{inf}$ decreases fast when the path loss factor $\eta$ increases. Compared with the set of results in Fig.~\ref{fig:Pinf_V}, the results in Fig.~\ref{fig:Pinf_V_large} vary in a much smaller range of the threshold virus volume $V_\text{th}$. When the number of infected individuals $N$ decreases, the density of infected individuals in this network increases, and the $\mathbb{P}_\text{inf}$ increases.

\subsection{Infectious Probability with Mobile Individuals}
We next present the numerical results when the infected individuals are moving according to the RWK model and RWP model. Kindly note that the results of infected individuals moving with the RD model are the same as the results of the static individuals because the mobility has no impact on the infectious probability as shown in Section~\ref{subsubsec:RD}.
We present the results of infected individuals moving with the RWK model and RWP model as follows.

\subsubsection{Results with RWK model}

We first present the results of infectious probability when infected individuals are moving with the RWK model in a relatively small network area. Fig.~\ref{fig:RWK_N_D100} shows the results, where the network radius $D$ is assumed to be 100m and the moving distance of each infected individual is assumed to be 20m, the path loss factor  $\eta$ is set to be 2.5. From Fig.~\ref{fig:RWK_N_D100}, we find that the number of infected individuals $N$ is still one of the key factors to determine the infectious probability. In particular, when the threshold virus volume is $1 \times 10^{-3}$ and the number of infected individuals $N$ is 30, the infectious probability is nearly 0. However, when the number of infected individuals $N$ becomes 50 and the threshold virus volume remains to be $1 \times 10^{-3}$, the infectious probability sharply increases to be 1. 

We then present the result of infectious probability with a larger network area. Fig.~\ref{fig:RWK_N_D1000} presents the results. We can observe from Fig.~\ref{fig:RWK_N_D1000} that the infectious probability drops fast as the threshold virus volume $V_\text{th}$ increases. For example, when the number of infected individuals $N$ is 300 and $V_\text{th}$ is $3 \times 10^{-5}$, the infectious probability is 0.9. When $V_\text{th}$ is increased to $3.4 \times 10^{-5}$, the infectious probability drops to nearly 0. Comparing Fig.~\ref{fig:RWK_N_D1000} with Fig.~\ref{fig:RWK_N_D100}, we can conclude that the infection risk for susceptible individuals becomes lower in a larger network area with a lower density of infected individuals.

\begin{figure}[t]
\centering
\includegraphics[width=7.8cm]{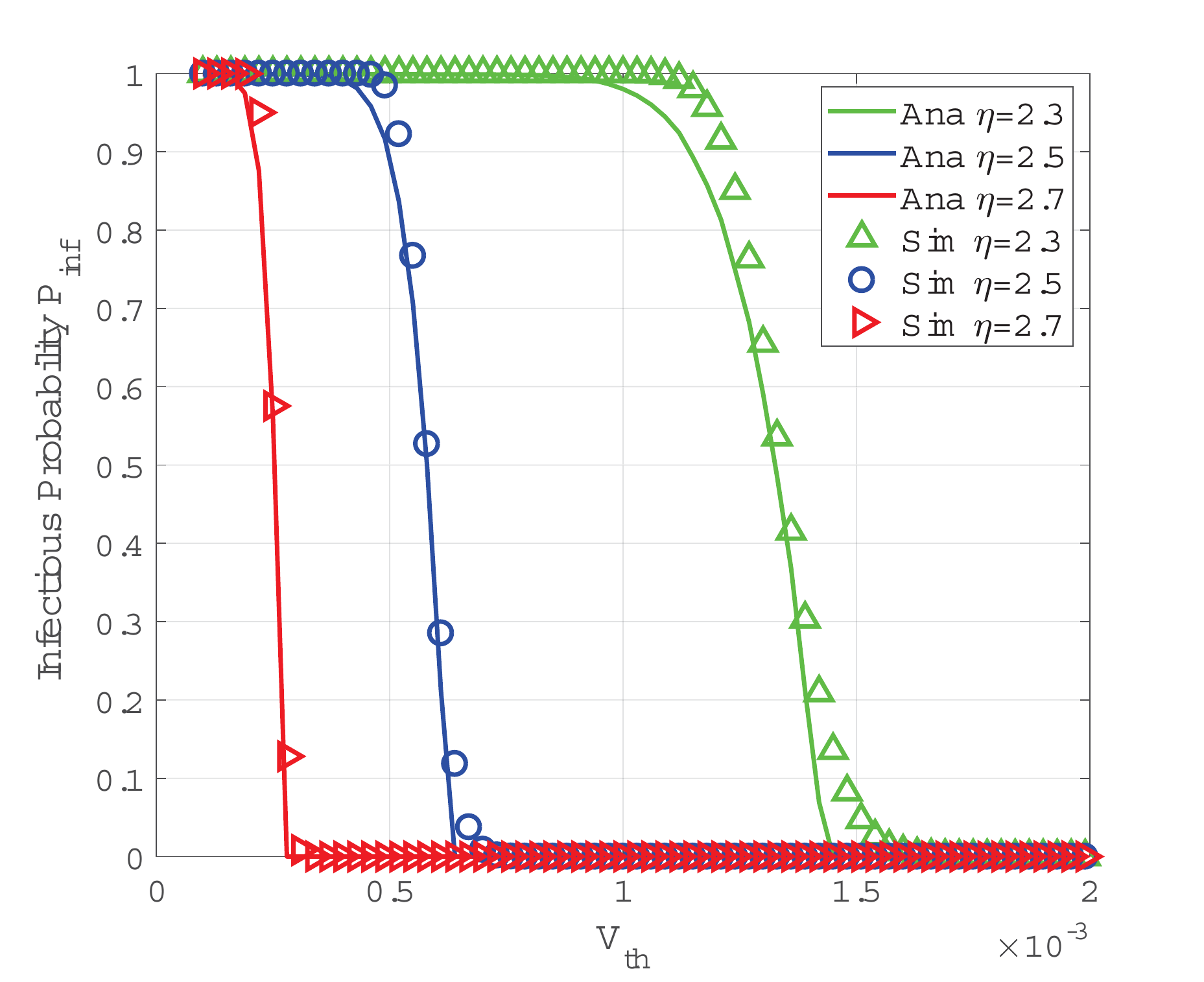}
\caption{Infectious probability $\mathbb{P}_\text{inf}$ versus the threshold virus volume ${V}_\text{th}$ with number of infected individuals $N = 20$, network radius $D=100$m and moving distance $W=20$m when infected individuals are moving with the RWK model.}
\label{fig:RWK_a_D100}
\end{figure}

We further investigate the impact of path loss factor $\eta$ on the infectious probability when infected individuals are moving with the RWK model. Fig.~\ref{fig:RWK_a_D100} plots the results. Comparing Fig.~\ref{fig:RWK_a_D100} with Fig. ~\ref{fig:Pinf_V-b}, where both two figures have the same number of infected individuals $N$ and network radius $D$, we find that the infectious probability in Fig.~\ref{fig:RWK_a_D100} is more sensitive to the path loss factor $\eta$. When the path loss factor $\eta$ increases from 2.3 to 2.5, the infectious probability significantly decreases. The reason for this phenomenon may lie in that the mobility leads to more contact chances (i.e., higher infectious risks) between the susceptible individuals and infected individuals, and the path loss factor $\eta$ may possibly determine whether the contact chances lead to the consequent infections to susceptible individuals.




\subsubsection{Results with RWP model}

We next present the results of the infectious probability $\mathbb{P}_\text{inf}$ with the RWP model. Fig.~\ref{fig:RWP_a_2to3} presents the infectious probability $\mathbb{P}_\text{inf}$ versus the threshold virus volume when the number of infected individuals $N$ is 20 and the radius of network area $D$ is 100. It is shown in Fig.~\ref{fig:RWP_a_2to3} that the infectious probability $\mathbb{P}_\text{inf}$ increases with the decreased path loss factor $\eta$. For instance, when the virus volume threshold $V_\text{th}$ is 0.01, $\mathbb{P}_\text{inf}$ is 0.105 when $\eta$ is $3$ though $\mathbb{P}_\text{inf}$ increases to be 0.819 when $\eta$ becomes $2$.

\begin{figure}[ht]
\centering
\includegraphics[width=7.4cm]{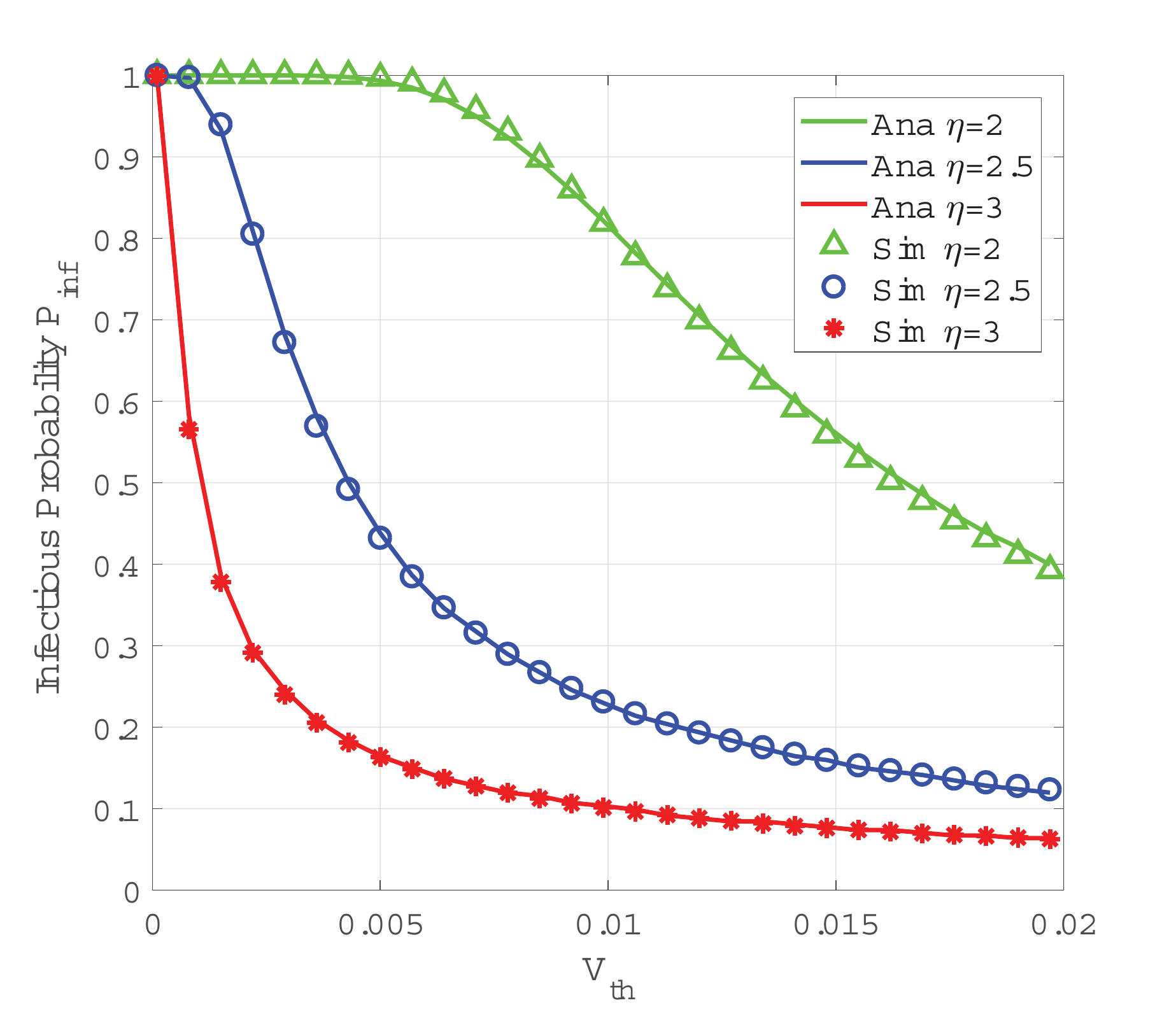}
\caption{Infectious probability $\mathbb{P}_\text{inf}$ versus the threshold virus volume ${V}_\text{th}$ with number of infected individuals $N = 20$ and network radius $D=100$m when infected individuals are moving with the RWP model.}
\label{fig:RWP_a_2to3}
\end{figure}

\begin{figure}[t]
\centering
\includegraphics[width=7.4cm]{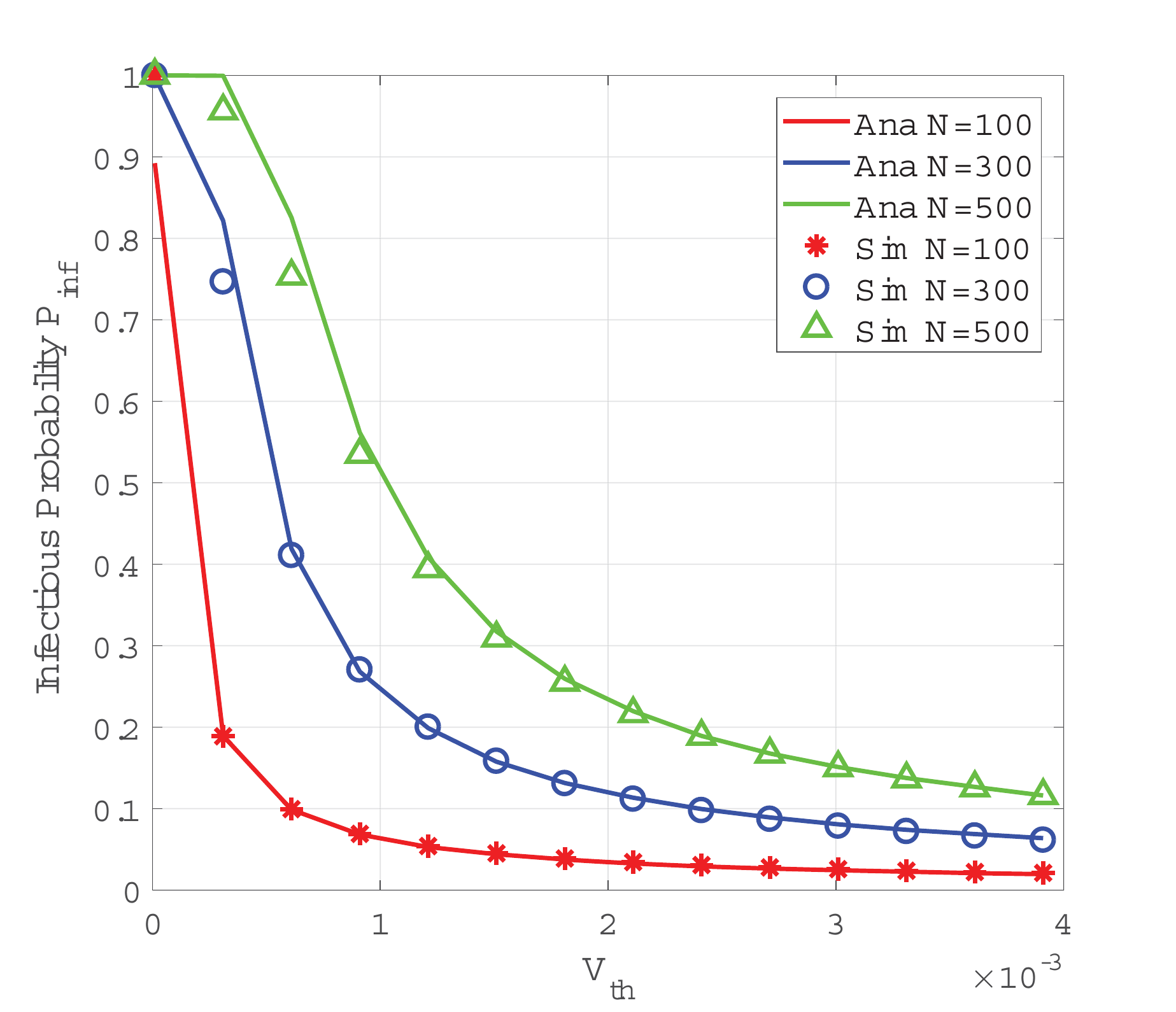}
\caption{Infectious probability $\mathbb{P}_\text{inf}$ versus the threshold virus volume ${V}_\text{th}$ with path loss factor $\eta = 2.5 $ and network radius $D=1000$m when infected individuals are moving with RWP model.}
\label{fig:RWP_N13500}
\end{figure}

\begin{figure}[t]
\centering
\includegraphics[width=7.4cm]{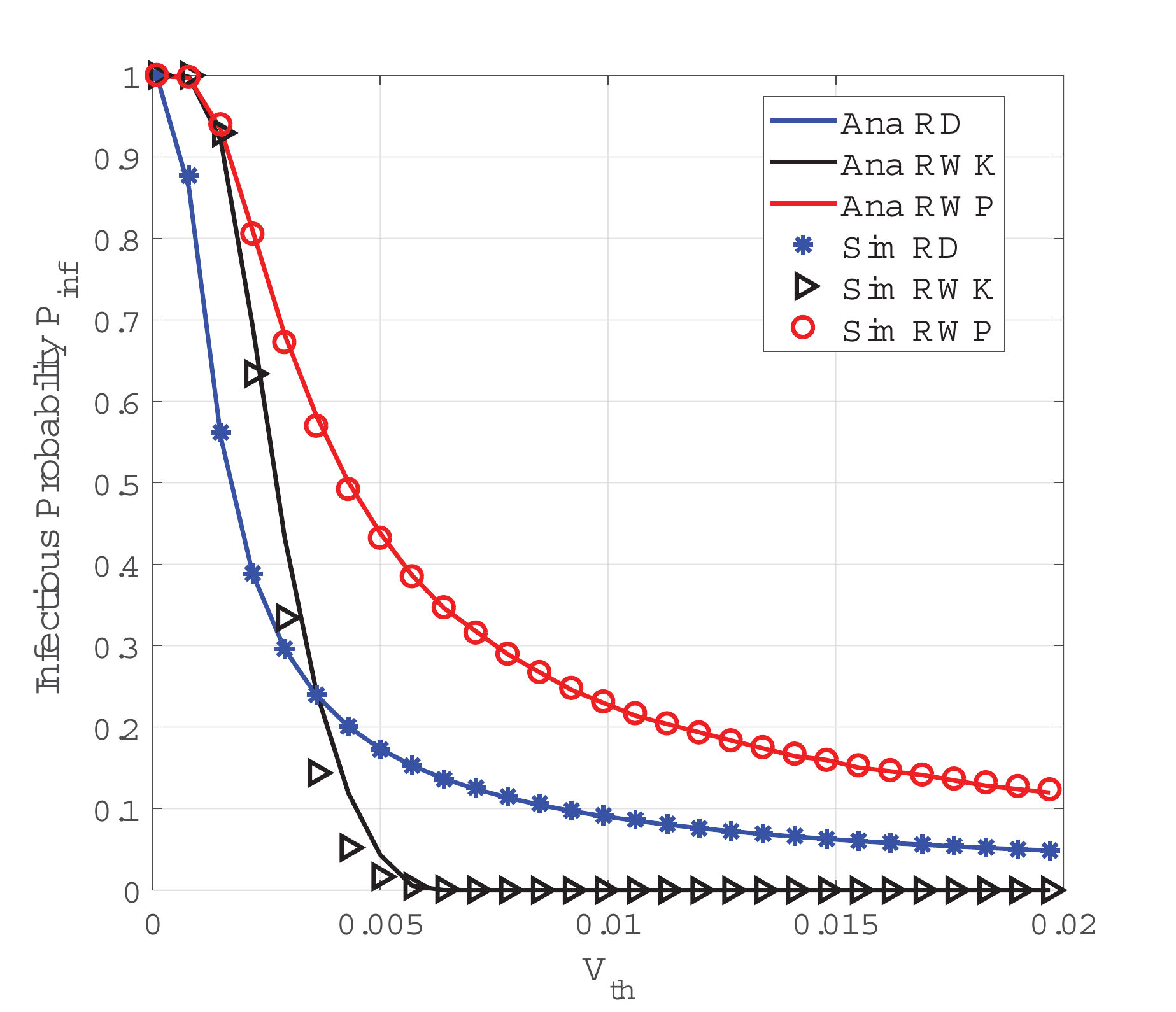}
\caption{Infectious probability $\mathbb{P}_\text{inf}$ versus the threshold virus volume ${V}_\text{th}$ with path loss factor $\eta = 2.5 $ and network radius $D=100$m when $20$ infected individuals ($N = 20$) are moving with three mobility models.}
\label{fig:Com_N20D100}
\end{figure}

Moreover, we present the results of the infectious probability when infected individuals are moving with the RWP model in a larger network area.
Fig.~\ref{fig:RWP_N13500} plots the results of the impact of infected individuals (i.e., $N$) on the infectious probability when the path loss factor $\eta = 2.5 $ and network radius $D=1000$m. We observe that the infectious probability $\mathbb{P}_\text{inf}$ increases when the number of infected individuals $N$ increases. For instance,
when the virus volume threshold $V_\text{th}$ is $1 \times 10^{-3}$ and the number of infected individuals $N$ is 100, the infectious probability $\mathbb{P}_\text{inf}$ is 0.064. However, $\mathbb{P}_\text{inf}$ becomes 0.512 when $N$ increases to be 500.
Comparing Fig.~\ref{fig:RWP_N13500} with Fig.~\ref{fig:RWP_a_2to3}, we find the infection risk of susceptible individuals in a larger network area with lower density is lower than that in a smaller network with high density.

\subsubsection{Results of Comparison}

\begin{figure}[t]
\centering
\includegraphics[width=7.4cm]{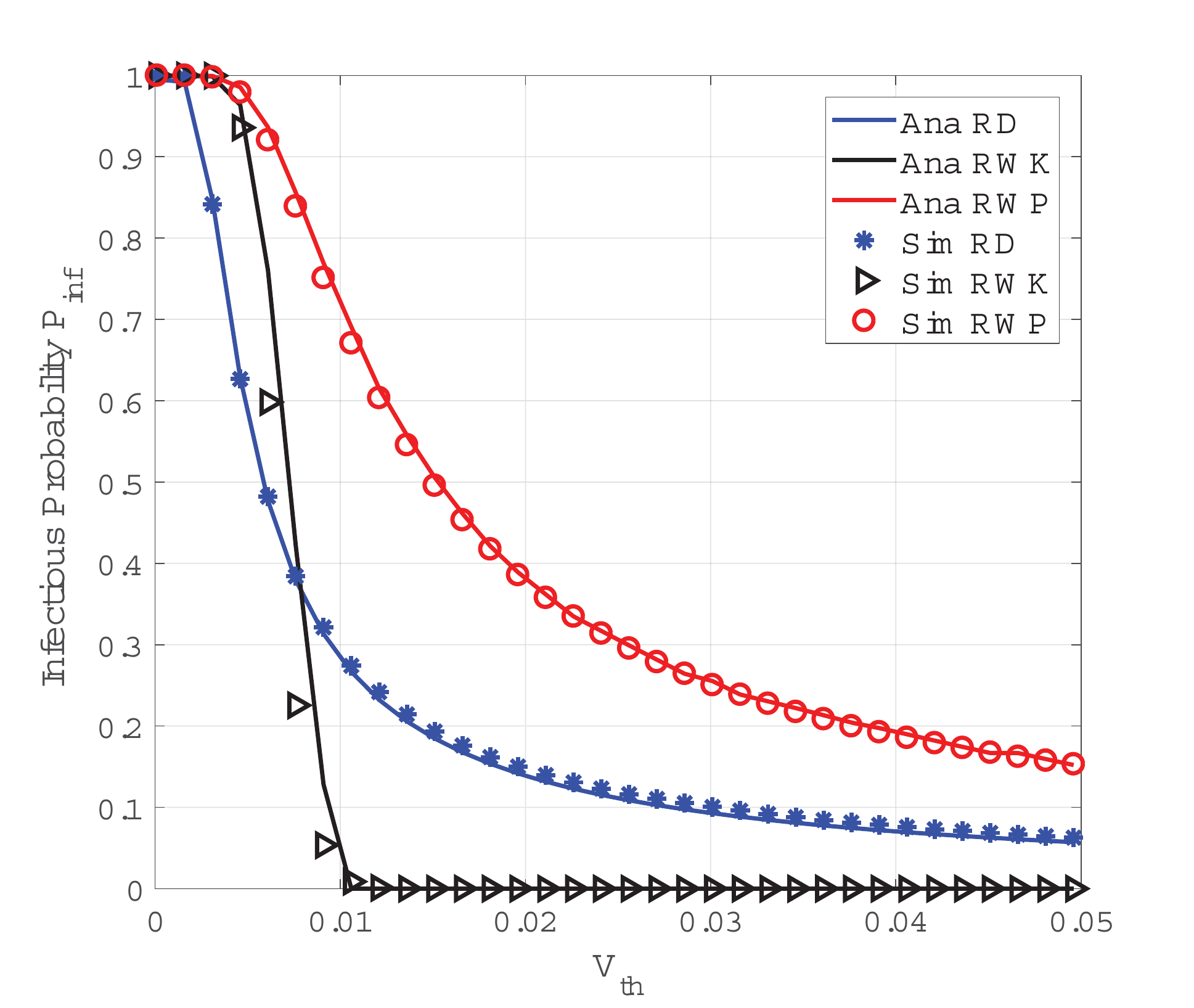}
\caption{Infectious probability $\mathbb{P}_\text{inf}$ versus the threshold virus volume ${V}_\text{th}$ with path loss factor $\eta = 2.5 $ and network radius $D=100$m when $50$ infected individuals ($N = 50$) are moving with three mobility models.}
\label{fig:Com_N50D100}
\end{figure}

We next compare the impacts of three mobility models on the infectious probability $\mathbb{P}_\text{inf}$ with path loss factor $\eta = 2.5 $ and network radius $D=100$m. Fig.~\ref{fig:Com_N20D100} presents the results of $\mathbb{P}_\text{inf}$ when 20 infected individuals are moving with the RD model, RWK model and RWP model.
As mentioned before, the results of infected individuals moving with the RD model are the same as the results of static individuals because the mobility has no impact on the infectious probability as shown in Section~\ref{subsubsec:RD}. Therefore, the green curve in Fig.~\ref{fig:Com_N20D100} also represents the results of static individuals.

Fig.~\ref{fig:Com_N20D100} shows that $\mathbb{P}_\text{inf}$ with the RWP model is always higher than that with the RD model. This result can be easily explained with the help of Fig.~\ref{fig:Trails-c} as given in Section~\ref{subsec:mobi-prob}. From the walking trails of the infected individual, we find that the distances between the infected individuals and the susceptible individual (located at the center of the circular region) are not uniformly distributed. More specifically, the infected individuals more likely move to be closer to the center of the circular region than to the boundary. Therefore, the infected individuals have higher contact chances to the susceptible individual. Moreover, we also observe that $\mathbb{P}_\text{inf}$ with the RWK model is higher than $\mathbb{P}_\text{inf}$ with the RD model when $V_\text{th}$ is smaller than $3.6 \times 10^{-3}$. However, $\mathbb{P}_\text{inf}$ with the RWK model is lower than that with the RD model when $V_\text{th}$ is larger than $3.6 \times 10^{-3}$. This phenomenon could probably be explained by the trails of the RWK model (given in Fig.~\ref{fig:Trails-b}) since the individual usually moves in a restricted area, which is situated between the center and the boundary of the circular network. In this case, the infected individuals have lower contact chances to the susceptible individual compared with the RWP model.

Fig.~\ref{fig:Com_N50D100} presents the results of $\mathbb{P}_\text{inf}$ when 50 infected individuals are moving with three mobility models.
Comparing Fig.~\ref{fig:Com_N50D100} with Fig.~\ref{fig:Com_N20D100}, we observe that $\mathbb{P}_\text{inf}$ with the RWP model is still higher than that with the RWP model when the number of infected individuals $N$ increases to 50.
Moreover, $\mathbb{P}_\text{inf}$ with the RWK model is higher than that with the RD model when $V_\text{th}$ is small (less than $7.3 \times 10^{-3}$). Similarly,  $\mathbb{P}_\text{inf}$ with the RWK model is higher than that with the RD model when $V_\text{th}$ is large.
When the the number of infected individuals $N$ increases from 20 to 50,  $\mathbb{P}_\text{inf}$ with all the three mobility models increases significantly.

\section{Countermeasures based on Wireless Edge Networks}
\label{sec:future}

The numerical results in Section~\ref{sec:result} offer many insightful implications. For example, the infectious probability $\mathbb{P}_\text{inf}$ is mainly influenced by the infectious model, the density of infected individuals, and the mobility models of individuals. Therefore, inspired by these findings, we can design countermeasures against the spread of COVID-19 based on wireless edge networks, as shown in Fig.~\ref{fig:countermeasures}. We summarize these countermeasures in the following three perspectives.

\begin{figure}[t]
\centering
\includegraphics[width=8.8cm]{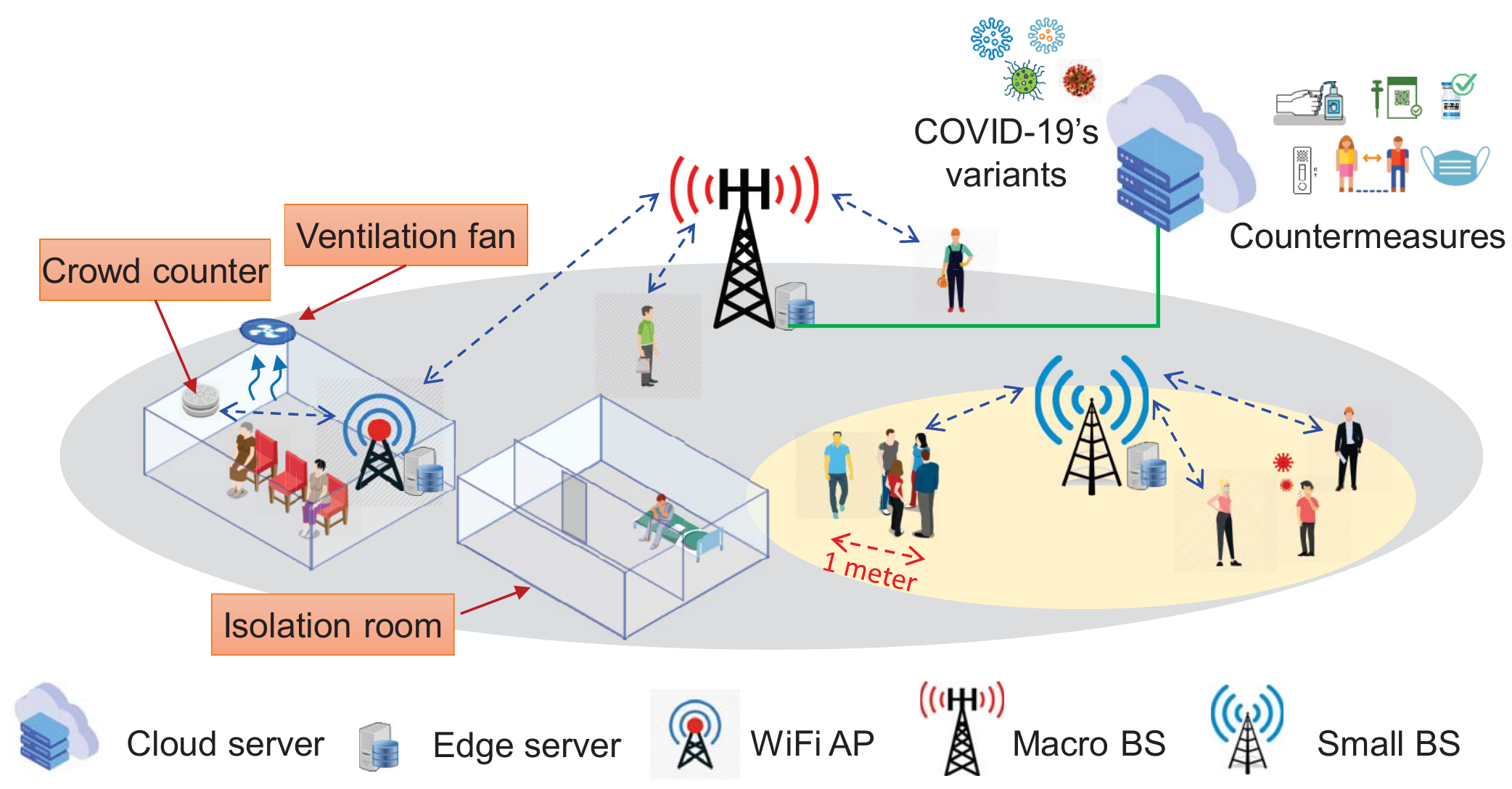}
\caption{Countermeasures based on Wireless Edge Networks}
\label{fig:countermeasures}
\end{figure}


\subsection{Informing public to keep social distance}

Our numerical results show that the infectious probability is mainly affected by multiple factors, such as the path loss factor of virus spreading, the crowd density, and the mobility models. For example, the shorter radius of the virus spreading region will lead to the higher infectious probability. Thus, as an effective countermeasure suggested by many previous studies~\cite{SUN2020102390}, keeping social distance is necessary, especially in the crowded environment.

With the aid of wireless edge networks, we can design a set of corresponding countermeasures to inform the public to keep social distance. For example, once the number of infected individuals is known by a nearby base station (e.g., a small base station in 4G/5G networks) and triggers the predefined threshold (e.g., 10 infected individuals), the edge server will broadcast notifications (e.g., sending text messages) to the public within this network so that the public can keep social distance (e.g., keeping a least 1 meter social distance) and wear surgery face masks.

Moreover, as shown previous studies~\cite{ABBOAHOFFEI2021100013,SUN2020102390}, COVID-19 and its variants have higher transmission risks in indoor environments or other places lacking enough ventilation. The integration of wireless edge networks with other services can potentially address this issue. In particular, the edge server deployed in approximation to the restaurant can automatically notify the central ventilation system to increase the speed of ventilation fans if the number of people in a room reaches a threshold (or being notified with an infected individual appearing in the restaurant). The number of people can be obtained by a crowd counter, which is connected to a nearby WiFi AP (or base station).


\subsection{Isolating infected individuals in time}

As shown in recent studies~\cite{chow2021care,tan2021curating}, in many population-dense countries and regions, there are much higher transmission risks of COVID-19 and its variants in nursing home than other places. Meanwhile, older adults living in nursing home often have higher fatalities than other healthy adults due to underlying comorbidity (i.e., the simultaneous existence of multiple diseases). Thus, it is crucial to protect older adults, especially in nursing home and other healthcare institutions.

The introduction of wireless edge networks can offer a potential solution to this issue. For example, once being informed with a number of infected individuals in nursing home, the edge server (located at the nursing home) can notify the healthcare workers to take necessary actions, e.g., isolating the infected individuals to negative pressure isolation rooms so as to reduce the risk of virus spreading.




\subsection{Orchestration with cloud computing to take new countermeasures}
There are a number of variants of SARS-CoV-2 since its outbreak in 2020. Each variant has different characteristics, such as fatalities, transmission risks, and immune escape, which pose the challenges in taking immediate and effective countermeasures against the spreading of the virus.

Wireless edge networks cannot work alone to address the challenges brought by new variants of COVID-19 due to the limitation of computing and storage capabilities of edge nodes. Thus, the orchestration of edge computing with cloud computing becomes a necessity to combat COVID-19 and its variants. Since cloud servers have powerful storage capacity, real-world epidemiological data (available through World Health Organization and CDC) can be saved at remote cloud servers. When the infectious model of a new variant is available (done by epidemiologists), CDC and other departments may make new countermeasures (policies) against the spreading of this new variant. Remote cloud servers may distribute both infectious models and new policies to edge nodes so that new countermeasures can be immediately made to different communities (e.g., older adults in nursing home).



\balance

\section{Conclusion}
\label{sec:conclusion}

We have experienced the pandemic of COVID-19 around the world. This paper aims to provide precautionary measures against COVID-19 and other infectious diseases with the aid of wireless edge networks, which are essentially an integration of edge computing with wireless networks. In particular, we present an analytical framework to predict the infectious probability of infectious diseases with the assistance of wireless edge networks. Motivated by previous studies on transmission probabilities of infectious diseases, we propose a stochastic geometry-based method to analyze the infectious probability of individuals within wireless edge networks due to the availability of the recorded detention time of individuals and the density of infectious individuals and susceptible individuals. Compared with other methods that require the locations or trajectories of users, the proposed method can better protect the privacy of individuals since only the recorded detention time of individuals and the density of individuals in a network are required. Moreover, our analytical framework also considers three types of mobility models, thereby being closer to realistic scenarios. Extensive numerical results show that analytical results well match simulation results, consequently validating the accuracy of the proposed model. In addition, we also offer a number of countermeasures against the spread of COVID-19 based on wireless edge networks, including notification to keep social distance, isolation of infected individuals, and orchestration with cloud computing. We believe that the in-depth integration of other technologies, such as AI and big data analytics with wireless edge networks can eventually combat the pandemic of COVID-19.



\begin{appendices}
\section{Proof of Lemma 4}

Inspired by the previous study~\cite{Hyytia:2006TMC},
we adopt a similar approach to derive the simplified probability density function $f_{L }(l)$ of distance $l$ between the infected individual and susceptible individual.
In this appendix, we consider a circular area with radius $D$ where a susceptible individual is located at the center of this circular area $O$. Meanwhile, the infected individual is moving within this area according to the RWK model from the original point $ P_1$ to the endpoint $ P_2$. According to the RWK model, the length of line segment $ \overline{ P_1 P_2} $ is a fixed value $W$. Since the original location of the infected individual is uniformly distributed in the circular area, the length $z$ of line segment $ \overline{ P_1 O} $ is given by the following equation,
\begin{equation}
f_{Z}(z)=\frac{2 z}{D^{2}}, 0 < z \leq D.
\label{eq:fz}
\end{equation}

The endpoint $ P_2$ will be located on the boundary of a circular area, which is centered at the original point $ P_1$ with radius $W$. In the moving process from$ P_1$ to $ P_2$, the location of the infected individual may be at each point of this circular area since the moving direction is randomly selected within $[0,2 \pi]$.
We then derive the expressions of the probability density function $f_{L }(l)$ of the distance $l$ according to two different cases of $z$ as follows.

\emph{Case 1:} $W \leq z \leq D$.

When the length $z$ is determined, we can calculate the conditional CDF of $l$,
\begin{equation}
F_{L \mid Z}(l\mid z) = \frac{S}{\pi W^2},
\label{eq:FLZ}
\end{equation}
where $S$ is the target region satisfying the following condition: when the infected individual moves into the target area, the distance between this infected individual and the susceptible individual is smaller than $l$. Therefore, we derive the cumulative probability that the distance between the infected individual and the susceptible individual is less than $l$ with (\ref{eq:FLZ}); this is consistent with the definition of CDF.

According to the geometrical relationships of $W$, $l$ and $Z$, we can calculate the area of target region.
When $z-W \leq l<z$, the area of target region is derived by
\begin{equation}
\begin{aligned}
S&=(W^{2}\theta_{1}- W \sin \theta_{1} W \cos \theta_{1}) +( l^2  \theta_{2} - l \sin \theta_{2} l \cos \theta_{2})\\
&=W^{2}\theta_{1}+l^2  \theta_{2}- W z \sin \theta_{1},
\end{aligned}
\nonumber
\end{equation}
where $\theta_{1}$ equals to $\arccos \left[(W^{2}+z^{2}-l^{2})/(2 W z)\right]$, and $\theta_{2}$ equals to $\arccos \left[(l^{2}+z^{2}-W^{2})/(2 l z)\right]$.

Similarly, when $z \leq l<z+W$, the area of target region is calculated by
\begin{equation}
\begin{aligned}
S=&\pi W^{2}-\left[\left(W^{2}\left(\pi-\theta_{1}\right)-W \sin (\pi- \theta_{1}) \cdot W \cos \left(\pi-\theta_{1}\right)\right) \right.\\
 & \qquad\quad \left. -\left(l^{2} \theta_{2}-l \sin \theta_{2} \cdot l \cos \theta_{2}\right)\right] \\
=& W^{2} \theta_{1}+l^{2} \theta_{2}-W z \sin \theta_{1}.
\end{aligned}
\nonumber
\end{equation}

Summarizing the above equations, we have the following expression of $F_{L \mid Z}(l \mid z)$ as follows,

{\small
\begin{equation}
F_{L \mid Z}(l \mid z)=\left\{\begin{array}{l}
 \displaystyle 0,  \qquad\qquad\qquad\qquad\quad\qquad 0 \leq l<z-W \\
 \displaystyle\frac{W^{2} \theta_{1}+l^{2} \theta_{2}- z W \sin \theta_{1}}{\pi W^{2}},\ z-W \leq l<z+W \\
 \displaystyle 1, \qquad\qquad\qquad\qquad\qquad\quad z+W \leq l<D,
\end{array}\right.
\nonumber
\end{equation}
}

From the definition of Conditional CDF, we have
\begin{equation}
F_{L \mid Z}(l \mid z)=\frac{\int_{-\infty}^{l} f(l, z) d l}{f_{Z}(z)}.
\label{eq:jointFcase1}
\end{equation}

After the transformation of (\ref{eq:jointFcase1}), we can calculate the joint PDF of $l$ and $z$ as follows,
\begin{equation}
f(l, z)=\frac{\partial\left(F_{L \mid Z}(l \mid z) f_{Z}(z)\right)}{\partial l},
\nonumber
\end{equation}
where $f_{Z}(z)$ is given by (\ref{eq:fz}).

According to the property of joint PDF, we have
\begin{equation}
f_{L}(l)=\int_{-\infty}^{+\infty} f(l, z) d z =\int_{W}^{D} f(l, z) d z.
\label{eq:flcase1}
\end{equation}

We next derive the CDF by integrating PDF with $l$,
\begin{equation}
F_{L}(l)=\int_{-\infty}^{+\infty} f_{L}(l) d l  =\int_{z-W}^{D} f_{L}(l) d l.
\label{eq:Flcase1}
\end{equation}


\emph{Case 2}: $0 \leq z < W$.

Adopting the same approach to Case 1, we derive the expression of $F_{L \mid Z}'(l \mid z)$ as follows,

{\small
\begin{equation}
 \displaystyle F_{L \mid Z}'(l \mid z)=\left\{\begin{array}{l}
 \displaystyle \frac{ l^{2}}{W^{2}} ,  \qquad\qquad\qquad\qquad\qquad 0 \leq l<W-z \\
 \displaystyle\frac{ W^{2} \theta_{1}+l^{2} \theta_{2}- z W \sin \theta_{1}}{\pi W^{2}}, \ W-z \leq l< W+z \\
 \displaystyle 1,  \qquad\qquad\qquad\qquad\qquad\quad W+z \leq l<D.
\end{array}\right.
\nonumber
\end{equation}
}

Following the same derivation process in Case 1, we can derive the expressions of $f(l, z)$, $f_{L}(l)$ and $F_{L}(l)$ on condition of $0 \leq z < W$ .

The joint PDF of $l$ and $z$ is derived as follows,
\begin{equation}
f(l, z)'=\frac{\partial(F_{L \mid Z}'(l \mid z) f_{Z}(z))}{\partial l},
\nonumber
\end{equation}

The expression of PDF of $l$ is give by
\begin{equation}
f_{L}'(l)=\int_{o}^{W} f'(l, z) d z.
\label{eq:flcase2}
\end{equation}

Then we have the CDF of $l$ as follows,
\begin{equation}
F_{L}'(l)=\int_{0}^{D} f_{L}'(l) d l.
\label{eq:Flcase2}
\end{equation}

After the integration and summary of (\ref{eq:flcase1}) and (\ref{eq:flcase2}), we get the expression of $f_{L}(l)$ in (\ref{eq:kpdf}). Similarly, after the integration and summary of (\ref{eq:Flcase1}) and (\ref{eq:Flcase2}) , we get the expression of and the expression of $F_{L}(l)$ in (\ref{eq:kcdf}) as given in Lemma~\ref{lemma:kpdf}.

\end{appendices}

\bibliography{IEEEabrv,health}

\begin{thebibliography}{10}
\providecommand{\url}[1]{#1}
\csname url@rmstyle\endcsname
\providecommand{\newblock}{\relax}
\providecommand{\bibinfo}[2]{#2}
\providecommand\BIBentrySTDinterwordspacing{\spaceskip=0pt\relax}
\providecommand\BIBentryALTinterwordstretchfactor{4}
\providecommand\BIBentryALTinterwordspacing{\spaceskip=\fontdimen2\font plus
\BIBentryALTinterwordstretchfactor\fontdimen3\font minus
  \fontdimen4\font\relax}
\providecommand\BIBforeignlanguage[2]{{%
\expandafter\ifx\csname l@#1\endcsname\relax
\typeout{** WARNING: IEEEtran.bst: No hyphenation pattern has been}%
\typeout{** loaded for the language `#1'. Using the pattern for}%
\typeout{** the default language instead.}%
\else
\language=\csname l@#1\endcsname
\fi
#2}}

\bibitem{Wu:Nature2020}
F.~Wu \emph{et~al.}, ``A new coronavirus associated with human respiratory
  disease in {China},'' \emph{Nature}, vol. 579, pp. 265--269, Apr. 2020.

\bibitem{wong2022transmission}
S.-C. Wong, A.~K.-W. Au, H.~Chen, L.~L.-H. Yuen, X.~Li, D.~C. Lung, A.~W.-H.
  Chu, J.~D. Ip, W.-M. Chan, H.-W. Tsoi, \emph{et~al.}, ``{Transmission of
  Omicron (B. 1.1. 529)-SARS-CoV-2 Variant of Concern in a designated
  quarantine hotel for travelers: a challenge of elimination strategy of
  COVID-19},'' \emph{The Lancet Regional Health--Western Pacific}, vol.~18,
  Jan. 2022.

\bibitem{Chang:2020nature}
S.~Chang, E.~Pierson, W.~K. Pang, J.~Gerardin, and J.~Leskovec, ``Mobility
  network models of {COVID-19} explain inequities and inform reopening,''
  \emph{Nature}, pp. 1--6, Nov. 2020.

\bibitem{9405303}
M.~T. Barros, M.~Veletić, M.~Kanada, M.~Pierobon, S.~Vainio, I.~Balasingham,
  and S.~Balasubramaniam, ``Molecular communications in viral infections
  research: Modeling, experimental data, and future directions,'' \emph{IEEE
  Transactions on Molecular, Biological and Multi-Scale Communications},
  vol.~7, no.~3, pp. 121--141, 2021.

\bibitem{Ning:2021JSAC}
Z.~Ning, P.~Dong, X.~Wang, X.~Hu, L.~Guo, B.~Hu, Y.~Guo, T.~Qiu, and R.~Y.~K.
  Kwok, ``{Mobile Edge Computing Enabled 5G Health Monitoring for Internet of
  Medical Things: A Decentralized Game Theoretic Approach},'' \emph{IEEE
  Journal on Selected Areas in Communications}, vol.~39, no.~2, pp. 463--478,
  Feb. 2021.

\bibitem{Azana:2021JNCA}
A.~H. {Mohd Aman}, W.~H. Hassan, S.~Sameen, Z.~S. Attarbashi, M.~Alizadeh, and
  L.~A. Latiff, ``{IoMT amid COVID-19 pandemic: Application, architecture,
  technology, and security},'' \emph{Journal of Network and Computer
  Applications}, vol. 174, p. 102886, Jan. 2021.

\bibitem{Pace:2019TII}
P.~Pace, G.~Aloi, R.~Gravina, G.~Caliciuri, G.~Fortino, and A.~Liotta, ``{An
  Edge-Based Architecture to Support Efficient Applications for Healthcare
  Industry 4.0},'' \emph{IEEE Transactions on Industrial Informatics}, vol.~15,
  no.~1, pp. 481--489, Jan. 2019.

\bibitem{Qadri:CST20}
Y.~A. {Qadri}, A.~{Nauman}, Y.~B. {Zikria}, A.~V. {Vasilakos}, and S.~W. {Kim},
  ``{The Future of Healthcare Internet of Things: A Survey of Emerging
  Technologies},'' \emph{IEEE Communications Surveys \& Tutorials}, vol.~22,
  no.~2, pp. 1121--1167, Feb. 2020.

\bibitem{Zhou:2021JSAC}
Z.~Zhou, Z.~Wang, H.~Yu, H.~Liao, S.~Mumtaz, L.~Oliveira, and V.~Frascolla,
  ``{Learning-Based URLLC-Aware Task Offloading for Internet of Health
  Things},'' \emph{IEEE Journal on Selected Areas in Communications}, vol.~39,
  no.~2, pp. 396--410, Feb. 2021.

\bibitem{Hossain:2020IN}
M.~S. Hossain, G.~Muhammad, and N.~Guizani, ``{Explainable AI and Mass
  Surveillance System-Based Healthcare Framework to Combat {COVID-I9} Like
  Pandemics},'' \emph{{IEEE} Netw.}, vol.~34, no.~4, pp. 126--132, Aug. 2020.

\bibitem{Gomez:2021CogCom}
Gomez-Cravioto, D.A., Diaz-Ramos, R.E., Cantu-Ortiz, and F.J., ``{Data Analysis
  and Forecasting of the {COVID-19} Spread: A Comparison of Recurrent Neural
  Networks and Time Series Models},'' \emph{Cognitive Computation (Early
  Access)}, June 2021.

\bibitem{Islam:2021IA}
M.~M. Islam, F.~Karray, R.~Alhajj, and J.~Zeng, ``{A Review on Deep Learning
  Techniques for the Diagnosis of Novel Coronavirus (COVID-19)},'' \emph{IEEE
  Access}, vol.~9, pp. 30\,551--30\,572, Feb. 2021.

\bibitem{ren2022optimal}
J.~Ren, M.~Liu, Y.~Liu, and J.~Liu, ``Optimal resource allocation with
  spatiotemporal transmission discovery for effective disease control,''
  \emph{Infectious diseases of poverty}, vol.~11, no.~1, pp. 1--11, Mar. 2022.

\bibitem{LvW:2020TNCE}
W.~Lv, S.~Wu, C.~Jiang, Y.~Cui, X.~Qiu, and Y.~Zhang, ``Towards large-scale and
  privacy-preserving contact tracing in {COVID-19} pandemic: A blockchain
  perspective,'' \emph{IEEE Transactions on Network Science and Engineering},
  vol.~9, no.~1, pp. 282--298, Oct. 2020.

\bibitem{xrli:2021PMC}
X.~Li, B.~Tao, H.-N. Dai, M.~Imran, D.~Wan, and D.~Li, ``{Is blockchain for
  Internet of Medical Things a panacea for {COVID-19} pandemic?}''
  \emph{Pervasive and Mobile Computing}, vol.~75, p. 101434, Aug. 2021.

\bibitem{Fakhri2020:SCS}
F.~{Alam Khan}, M.~Asif, A.~Ahmad, M.~Alharbi, and H.~Aljuaid, ``Blockchain
  technology, improvement suggestions, security challenges on smart grid and
  its application in healthcare for sustainable development,''
  \emph{Sustainable Cities and Society}, vol.~55, p. 102018, Apr. 2020.

\bibitem{9665215}
A.~Mukherjee, P.~K. Deb, and S.~Misra, ``Tremors: Privacy-breaching inference
  of computing tasks using vibration-based condition monitors,'' \emph{IEEE
  Transactions on Computers}, vol. (early access), pp. 1--1, 2021.

\bibitem{10.1145/3425707}
\BIBentryALTinterwordspacing
S.~A. Chaudhry, A.~Irshad, K.~Yahya, N.~Kumar, M.~Alazab, and Y.~B. Zikria,
  ``Rotating behind privacy: An improved lightweight authentication scheme for
  cloud-based iot environment,'' \emph{ACM Transactions on Internet
  Technology}, vol.~21, no.~3, jun 2021. [Online]. Available:
  \url{https://doi.org/10.1145/3425707}
\BIBentrySTDinterwordspacing

\bibitem{9320508}
M.~A. Rahman and M.~S. Hossain, ``An internet-of-medical-things-enabled edge
  computing framework for tackling covid-19,'' \emph{IEEE Internet of Things
  Journal}, vol.~8, no.~21, pp. 15\,847--15\,854, Nov. 2021.

\bibitem{TOMAR2020138762}
\BIBentryALTinterwordspacing
A.~Tomar and N.~Gupta, ``Prediction for the spread of covid-19 in india and
  effectiveness of preventive measures,'' \emph{Science of The Total
  Environment}, vol. 728, p. 138762, Aug. 2020. [Online]. Available:
  \url{https://www.sciencedirect.com/science/article/pii/S0048969720322798}
\BIBentrySTDinterwordspacing

\bibitem{MORABITO2020140347}
\BIBentryALTinterwordspacing
M.~Morabito, A.~Messeri, A.~Crisci, L.~Pratali, M.~Bonafede, and A.~Marinaccio,
  ``{Heat warning and public and workers' health at the time of COVID-19
  pandemic},'' \emph{Science of The Total Environment}, vol. 738, p. 140347,
  Oct. 2020. [Online]. Available:
  \url{https://www.sciencedirect.com/science/article/pii/S0048969720338699}
\BIBentrySTDinterwordspacing

\bibitem{10.1145/3397271.3401429}
\BIBentryALTinterwordspacing
Z.~Fu, Y.~Wu, H.~Zhang, Y.~Hu, D.~Zhao, and R.~Yan, \emph{{Be Aware of the Hot
  Zone: A Warning System of Hazard Area Prediction to Intervene Novel
  Coronavirus COVID-19 Outbreak}}.\hskip 1em plus 0.5em minus 0.4em\relax New
  York, NY, USA: Association for Computing Machinery, Jul. 2020, p.
  2241–2250. [Online]. Available:
  \url{https://doi.org/10.1145/3397271.3401429}
\BIBentrySTDinterwordspacing

\bibitem{yan2008distribution}
P.~Yan, ``Distribution theory, stochastic processes and infectious disease
  modelling,'' in \emph{Mathematical epidemiology}.\hskip 1em plus 0.5em minus
  0.4em\relax Springer, 2008, pp. 229--293.

\bibitem{ZHANG2020201}
\BIBentryALTinterwordspacing
S.~Zhang, M.~Diao, W.~Yu, L.~Pei, Z.~Lin, and D.~Chen, ``{Estimation of the
  reproductive number of novel coronavirus (COVID-19) and the probable outbreak
  size on the Diamond Princess cruise ship: A data-driven analysis},''
  \emph{International Journal of Infectious Diseases}, vol.~93, pp. 201--204,
  Apr. 2020. [Online]. Available:
  \url{https://www.sciencedirect.com/science/article/pii/S1201971220300916}
\BIBentrySTDinterwordspacing

\bibitem{LIU2021106542}
\BIBentryALTinterwordspacing
F.~Liu, Z.~Luo, Y.~Li, X.~Zheng, C.~Zhang, and H.~Qian, ``Revisiting physical
  distancing threshold in indoor environment using infection-risk-based
  modeling,'' \emph{Environment International}, vol. 153, p. 106542, Aug. 2021.
  [Online]. Available:
  \url{https://www.sciencedirect.com/science/article/pii/S0160412021001677}
\BIBentrySTDinterwordspacing

\bibitem{Andrews:2011Tcom}
J.~G. Andrews, F.~Baccelli, and R.~K. Ganti, ``A tractable approach to coverage
  and rate in cellular networks,'' \emph{IEEE Transactions on Communications},
  vol.~59, no.~11, pp. 3122--3134, Oct. 2011.

\bibitem{Hmamouche:2021IP}
Y.~Hmamouche, M.~Benjillali, S.~Saoudi, H.~Yanikomeroglu, and M.~D. Renzo,
  ``New trends in stochastic geometry for wireless networks: A tutorial and
  survey,'' \emph{Proceedings of the IEEE}, vol. 109, no.~7, pp. 1200--1252,
  Mar. 2021.

\bibitem{Lu:2021CST}
X.~Lu, M.~Salehi, M.~Haenggi, E.~Hossain, and H.~Jiang, ``Stochastic geometry
  analysis of spatial-temporal performance in wireless networks: A tutorial,''
  \emph{IEEE Communications Surveys \& Tutorials}, vol.~23, no.~4, pp.
  2753--2801, Aug. 2021.

\bibitem{Govindan:2011TWC}
K.~Govindan, K.~Zeng, and P.~Mohapatra, ``{Probability Density of the Received
  Power in Mobile Networks},'' \emph{IEEE Transactions on Wireless
  Communications}, vol.~10, no.~11, pp. 3613--3619, Nov. 2011.

\bibitem{Tang:2020IoT}
J.~Tang, H.~Wen, H.~Song, T.~Zhang, and K.~Qin, ``{On the Security-Reliability
  and Secrecy Throughput of Random Mobile User in Internet of Things},''
  \emph{IEEE Internet of Things Journal}, vol.~7, no.~10, pp. 10\,635--10\,649,
  Oct. 2020.

\bibitem{XJunfei:2014ICST}
J.~Xie, Y.~Wan, J.~H. Kim, S.~Fu, and K.~Namuduri, ``{A Survey and Analysis of
  Mobility Models for Airborne Networks},'' \emph{IEEE Communications Surveys
  \& Tutorials}, vol.~16, no.~3, pp. 1221--1238, Dec. 2014.

\bibitem{GZhenhua:2011ICC}
Z.~Gong and M.~Haenggi, ``{Temporal Correlation of the Interference in Mobile
  Random Networks},'' in \emph{2011 IEEE International Conference on
  Communications (ICC)}, Jul. 2011, pp. 1--5.

\bibitem{TJie:2018TWC}
J.~Tang, M.~Dabaghchian, K.~Zeng, and H.~Wen, ``{Impact of Mobility on Physical
  Layer Security Over Wireless Fading Channels},'' \emph{IEEE Transactions on
  Wireless Communications}, vol.~17, no.~12, pp. 7849--7864, Dec. 2018.

\bibitem{Valentine:2020IA}
V.~A. Aalo, P.~S. Bithas, and G.~P. Efthymoglou, ``{On the Impact of User
  Mobility on the Performance of Wireless Receivers},'' \emph{IEEE Access},
  vol.~8, pp. 197\,300--197\,311, Oct. 2020.

\bibitem{Fernandez:2018TVT}
J.~Lopez-Fernandez, E.~Martos-Naya, F.~J. Lopez-Martinez, and T.~Tsiftsis,
  ``{On the Distribution of the Received Signal Power in Mobile Networks: A
  Moment-Based Approach},'' \emph{IEEE Transactions on Vehicular Technology},
  vol.~67, no.~8, pp. 7754--7758, Aug. 2018.

\bibitem{RIBEIRONAVARRETE2021120681}
\BIBentryALTinterwordspacing
S.~Ribeiro-Navarrete, J.~R. Saura, and D.~Palacios-Marqués, ``{Towards a new
  era of mass data collection: Assessing pandemic surveillance technologies to
  preserve user privacy},'' \emph{Technological Forecasting and Social Change},
  vol. 167, p. 120681, Jun. 2021. [Online]. Available:
  \url{https://www.sciencedirect.com/science/article/pii/S004016252100113X}
\BIBentrySTDinterwordspacing

\bibitem{SUN2020102390}
\BIBentryALTinterwordspacing
C.~Sun and Z.~Zhai, ``{The efficacy of social distance and ventilation
  effectiveness in preventing COVID-19 transmission},'' \emph{Sustainable
  Cities and Society}, vol.~62, p. 102390, Nov. 2020. [Online]. Available:
  \url{https://www.sciencedirect.com/science/article/pii/S2210670720306119}
\BIBentrySTDinterwordspacing

\bibitem{SINGANAYAGAM2022183}
\BIBentryALTinterwordspacing
A.~Singanayagam \emph{et~al.}, ``{Community transmission and viral load
  kinetics of the SARS-CoV-2 delta (B.1.617.2) variant in vaccinated and
  unvaccinated individuals in the UK: a prospective, longitudinal, cohort
  study},'' \emph{The Lancet Infectious Diseases}, vol.~22, no.~2, pp.
  183--195, Feb. 2022. [Online]. Available:
  \url{https://www.sciencedirect.com/science/article/pii/S1473309921006484}
\BIBentrySTDinterwordspacing

\bibitem{he2020temporal}
X.~He, E.~H. Lau, P.~Wu, X.~Deng, J.~Wang, X.~Hao, Y.~C. Lau, J.~Y. Wong,
  Y.~Guan, X.~Tan, \emph{et~al.}, ``Temporal dynamics in viral shedding and
  transmissibility of covid-19,'' \emph{Nature medicine}, vol.~26, no.~5, pp.
  672--675, Apr. 2020.

\bibitem{Bandyopadhyay:2007TMC}
S.~Bandyopadhyay, E.~J. Coyle, and T.~Falck, ``Stochastic properties of
  mobility models in mobile ad hoc networks,'' \emph{IEEE Transactions on
  Mobile Computing}, vol.~6, no.~11, pp. 1218--1229, Nov. 2007.

\bibitem{Hyytia:2006TMC}
E.~Hyytia, P.~Lassila, and J.~Virtamo, ``Spatial node distribution of the
  random waypoint mobility model with applications,'' \emph{{IEEE} Trans.
  Mobile Comput.}, vol.~5, no.~6, pp. 680--694, June 2006.

\bibitem{Martin:journalsG09}
\BIBentryALTinterwordspacing
M.~Haenggi and R.~K. Ganti, ``Interference in large wireless networks,''
  \emph{Found. Trends Netw.}, vol.~3, no.~2, pp. 127--248, 2009. [Online].
  Available: \url{https://doi.org/10.1561/1300000015}
\BIBentrySTDinterwordspacing

\bibitem{Gong:2014TMC}
Z.~Gong and M.~Haenggi, ``Interference and outage in mobile random networks:
  Expectation, distribution, and correlation,'' \emph{IEEE Transactions on
  Mobile Computing}, vol.~13, no.~2, pp. 337--349, Feb. 2014.

\bibitem{Irio:2018TCOM}
L.~Irio, A.~Furtado, R.~Oliveira, L.~Bernardo, and R.~Dinis, ``Interference
  characterization in random waypoint mobile networks,'' \emph{IEEE
  Transactions on Wireless Communications}, vol.~17, no.~11, pp. 7340--7351,
  Nov. 2018.

\bibitem{Chetlur:2017TCOM}
V.~V. Chetlur and H.~S. Dhillon, ``{Downlink Coverage Analysis for a Finite 3-D
  Wireless Network of Unmanned Aerial Vehicles},'' \emph{IEEE Transactions on
  Communications}, vol.~65, no.~10, pp. 4543--4558, Oct. 2017.

\bibitem{Alouini:1999TVT}
M.-S. Alouini and A.~Goldsmith, ``Area spectral efficiency of cellular mobile
  radio systems,'' \emph{IEEE Transactions on Vehicular Technology}, vol.~48,
  no.~4, pp. 1047--1066, Jul. 1999.

\bibitem{ABBOAHOFFEI2021100013}
\BIBentryALTinterwordspacing
M.~Abboah-Offei, Y.~Salifu, B.~Adewale, J.~Bayuo, R.~Ofosu-Poku, and E.~B.~A.
  Opare-Lokko, ``A rapid review of the use of face mask in preventing the
  spread of covid-19,'' \emph{International Journal of Nursing Studies
  Advances}, vol.~3, p. 100013, Nov. 2021. [Online]. Available:
  \url{https://www.sciencedirect.com/science/article/pii/S2666142X20300126}
\BIBentrySTDinterwordspacing

\bibitem{chow2021care}
L.~Chow, ``{Care homes and COVID-19 in Hong Kong: how the lessons from SARS
  were used to good effect},'' \emph{Age and ageing}, vol.~50, no.~1, pp.
  21--24, Jan. 2021.

\bibitem{tan2021curating}
M.~K.~B. Tan and C.~M. Tan, ``{Curating wellness during a pandemic in
  Singapore: COVID-19, museums, and digital imagination},'' \emph{Public
  Health}, vol. 192, pp. 68--71, Mar 2021.

\end{thebibliography}


\end{document}